\documentclass[12pt]{article}

\usepackage[margin=2.5cm]{geometry}

\usepackage{amsmath,amsthm,amsfonts,amscd,amssymb,bbm,mathrsfs, enumerate}

\usepackage{graphicx,tikz}
\usepackage[pdfborder={0 0 0}]{hyperref}
\usetikzlibrary{arrows}

\def\RR{{\mathbb R}}
\def\CC{{\mathbb C}}

\def\A{{\mathcal A}}

\def\D{{\mathcal D}}

\def\H{{\mathcal H}}

\def\K{{\mathcal K}}
\def\M{{\mathcal M}}

\def\P{{\mathcal P}}
\def\R{{\mathcal R}}

\def\a{\alpha}
\def\b{\beta}

\def\L{{\mathrm L}}

\def\R{{\mathrm R}}

\def\Ad{{\hbox{\rm Ad\,}}}

\def\1{{\mathbbm 1}}

\def\u1net{{\A^{(0)}}}

\def\diff{{\rm Diff}}

\def\diffs1{\diff(S^1)}
\def\mob{{\rm M\ddot{o}b}}
\def\mob2{{\rm M\ddot{o}b}^{(2)}}

\def\supp{{\rm supp\,}}

\def\psl2r{{\rm PSL}(2,\RR)}
\def\sl2r{{\rm SL}(2,\RR)}
\def\su11{{\rm SU}(1,1)}
\def\2dmob{{\overline{\psl2r}\times\overline{\psl2r}}}
\def\<{\langle}
\def\>{\rangle}

\def\Im{\mathrm{Im}\,}
\def\re{\mathrm{Re}\,}
\def\im{\mathrm{Im}\,}
\def\res{\mathrm{Res}\,}
\def\blaschke{\mathrm{Blaschke}\,}

\def\poincare{{\P^\uparrow_+}}

\def\dom{{\mathrm{Dom}}}
\def\fct{\widetilde{\phi}}

\newcommand{\Om}{\Omega}

\newcommand{\epslim}{\lim_{\epsilon \searrow 0}}

\newtheorem{theorem}{Theorem}[section]

\newtheorem{proposition}[theorem]{Proposition}
\newtheorem{lemma}[theorem]{Lemma}
\theoremstyle{remark}

\title{Wedge-local fields in integrable models with bound states}
\date{}
\author{
{\bf Daniela Cadamuro} \\
e-mail: {\tt dc13950@bristol.ac.uk}\\
Department of Mathematics, University of Bristol \\
University Walk, Bristol BS8 1TW, United Kingdom \\
{\bf Yoh Tanimoto}\\
e-mail: {\tt hoyt@ms.u-tokyo.ac.jp}\\
Graduate School of Mathematical Sciences, The University of Tokyo\\
and Institut f\"ur Theoretische Physik, G\"ottingen University\\
3-8-1 Komaba Meguro-ku Tokyo 153-8914, Japan.\\
JSPS SPD postdoctoral fellow\\
}
\begin{document}
\maketitle
\begin{abstract}
Recently, large families of two-dimensional quantum field theories with
factorizing S-matrices have been constructed by the operator-algebraic
methods, by first showing the existence of observables localized in wedge-shaped regions.
However, these constructions have been limited to the class of S-matrices
whose components are analytic in rapidity in the physical strip.

In this work, we construct candidates for observables in wedges for scalar factorizing S-matrices with poles in the physical strip
and show that they weakly commute on a certain domain.
We discuss some technical issues concerning further developments, especially the self-adjointness
of the candidate operators here and strong commutativity between them.
\end{abstract}

\section{Introduction}\label{introduction}
In recent years, we have seen many interesting developments in constructing models of
Quantum Field Theory (QFT) in the operator-algebraic approach \cite{Haag96}.
Among them, one of the most important contributions was the construction of $1+1$ dimensional
scalar quantum field theories with factorizing S-matrices \cite{Lechner03, Lechner08}.
In that work, Lechner took a large family of analytic functions which satisfy certain conditions
and constructed quantum field theories which have these functions as the two-particle
S-matrix. This was remarkable because, although physicists conjectured several properties
and computed many interesting quantities, a mathematically consistent construction of the models
in axiomatic approaches had not been obtained before, and also because these models include many S-matrices which had not
been considered by physicists, as the arguments exploit only a few properties of the S-matrices and
do not depend on specific expressions or computations.
Now, it was clear at the first step of the construction \cite{Lechner03} that this
method does not apply directly if the function $S(\zeta)$ of the input has simple poles 
in the physical strip, $0< \im \zeta <\pi$.
On the other hand, some integrable QFTs are believed to have S-matrices with simple poles
and they should correspond to bound states of elementary particles.
The purpose of this paper is to extend the step of \cite{Lechner03} to the cases with
poles in the physical strip.

In the models constructed in \cite{Lechner03, Lechner08}, if $n$ particles are incoming, then
after the scattering process $n$ particles are outgoing. Furthermore, the scattering process
of $n$ particles can be systematically constructed from that of $2$ particles. In this case,
the S-matrix is said to be factorizing. This property is expected for so-called integrable
models, which have infinitely many conserved currents.
In some cases, for example the Sine-Gordon model, Fr\"ohlich and Seiler constructed
Euclidean Green's functions \cite{FS76} and proved that the S-matrix is nontrivial.
But, to the authors' knowledge, it is unknown whether it is factorizing.
More general integrable models have more complicated Lagrangians, and the understanding
of these models is limited to the perturbation theory.

An alternative approach to integrable models is the form factor program \cite{Smirnov92, BK04}.
In this program, instead of quantizing the fields, one conjectures the S-matrix from
the symmetry of the Lagrangian. Then the matrix components of local fields (form factors) are
computed and then the $n$-point function should be reconstructed from these matrix components.
Although many interesting quantities have been computed, a convergence proof of
the form factor expansion of the $n$-point functions is still lacking in almost all cases
\cite{BK04}.

The recent operator-algebraic approach, initiated by Schroer \cite{Schroer99}, considers
so-called wedge-local fields, rather than point-like quantum fields. Wedge-local fields
are observables localized in an infinitely extended, wedge-shaped region.
Indeed, this infinite extension allows one to take operators which are very simple in
the momentum space \cite{Lechner03, LS14}.
Thereafter, the algebra of local, finitely extended observables
should be defined as the intersection of two algebras corresponding to left- and right-wedges \cite{Lechner08}.
It is noteworthy that the existence proof of local observables avoids explicit computations
of quantum fields, but it is reduced to a certain phase space property called modular
nuclearity \cite{BDL90, BL04}. The resulting net of local observables reproduces
the S-matrix of the input, therefore the inverse problem for a class of S-matrices has been solved.

Now, let us recall that the S-matrices treated in \cite{Lechner03, Lechner08} must have
analytic components in the physical strip. In the operator-algebraic approach, there have been
many other constructions of algebras
corresponding to wedges \cite{Alazzawi13, BT13, BT15, BLS11, DT11, GL07, Lechner12, LST13, Tanimoto12-2, Tanimoto14-1}
and some related constructions \cite{LW11, Bischoff12, LL15} which take an analytic function as an input,
but it was always assumed that the function has no pole in a certain region.
On the other hand, a pole in the physical strip of the S-matrix components
is considered to correspond to a bound state.
Indeed, some integrable models, e.g.\! the Sine-Gordon model, are believed to have S-matrix
with poles in the physical strip. These S-matrices have not yet been treated in the operator-algebraic approach.
For an S-matrix $S$ without poles, Lechner constructed a pair of operator-valued distrubitions $\phi, \phi'$
and proved that they commute when they are smeared by test functions supported in the left- and right-wedges,
respectively. This computation involves a shift of the integral contour of a function containing $S$.
If one takes the same construction for $S$ with poles in the physical strip, the shift of the integral contour 
yields the residues of $S$ at these poles. Hence, the fields $\phi, \phi'$ themselves cannot be wedge-local.

The same problem has appeared also in the form factor program \cite{BK04}. Solutions of the form factor equations should represent
the matrix components of local operators, and their commutators should vanish when they are spacelike separated.
Indeed, such a formal proof has been given first by Smirnov for S-matrices without poles \cite{Smirnov92}.
A crucial part of the proof is again done by shifting integral contours, which is invalid when the S-matrices
have poles. For S-matrices with poles, they added further properties to form factors which make a correspondence between the poles
and ``bound states''. Quella showed that these new properties allow one to cancel the residues which come from
first-order poles \cite{Quella99}. Higher poles are discussed by Babujian, Foerster and Karowski \cite{BFK06}
and local commutativity appears to be formally maintained, at least for some specific models (the so-called $Z(N)$-Ising models).

In this paper, for a certain S-matrix with poles in the physical strip, we introduce
a new field $\fct = \phi + \chi$, by adding an operator $\chi$, which we call the bound-state operator, to the field $\phi$ of Lechner.
It is the commutator of $\chi$ with its reflected operator $\chi'$ which cancels the contribution of the residues 
coming from the commutator between $\phi$ and $\phi'$, mentioned above.
The operator $\chi$ has a formal integral expression in terms of Zamolodchikov-Faddeev operators $z,z^\dagger$ with complex arguments.

To prove wedge-commutativity in operator-algebraic approach, one considers the reflected field $\fct'$ using the action of the CPT operator and
shows that its commutator with $\fct$ vanishes for test functions supported in right- and left-wedges, respectively,
in the weak sense, namely, the matrix elements of the commutator between suitable vectors vanish.
Strong commutativity remains to be proven.
Yet, by assuming the existence of nice self-adjoint extensions, we give an argument for the Reeh-Schlieder property.
We also argue that our wedge-local fields are non-temperate polarization-free generators \cite{BBS01}.
Besides, we classify the scalar S-matrices with poles in the physical strip that satisfy the requirements of our analysis.
The general form of the S-matrix that we obtain is essentially given by a certain subclass of S-matrices
known from \cite{Lechner:2006} multiplied with a universal model independent factor which has poles in the physical strip.
This especially includes the S-matrices of the Bullough-Dodd model \cite{FMS93}.

Further, we clarify how the (weak) wedge localization of the field $\fct$ is related to the properties
of form factors \cite{BabujianFoersterKarowski:2006}. We conjecture a generalization of
the characterization theorem of local observables in \cite{BC14} to models with bound states.
In \cite{BC14} local observables are expanded into a series in terms of Zamolodchikov-Faddeev operators,
where the expansion coefficients are related to the form factors of the observable.
We sketch the outline of a proof that,
if its expansion coefficients of an operator fulfill a set of conditions, slightly modified from \cite{BC14},
then it formally commutes with our fields $\fct(f)$, hence they are local to each other.

The paper is organized as follows:
In Section \ref{sec:preliminaries} we recall the results of Lechner \cite{Lechner03, Lechner08} for S-matrices
analytic in the physical strip, and we introduce our general notation.
In Section \ref{sec:scalar} we summarize the properties of scalar S-matrices with poles in the physical strip
in the models under investigation, and we construct the wedge-local fields $\fct, \fct'$. A proof of weak
wedge-commutativity is also given in this section.
In Section \ref{sec:formfactor} we show that our wedge-local fields are compatible with the form factor program
in the sense explained above. Further, we explain how the form of the bound-state operator can be deduced
by formal arguments.
In Section \ref{sec:conclusions} we present our conclusions and open problems.
Appendix \ref{classification} is dedicated to the classification of scalar S-matrices fulfilling
the properties introduced in Section \ref{sec:scalar}.
Appendix \ref{domain} comments on the problem of finding suitable self-adjoint extensions of the field $\fct$.

\section{Preliminaries}\label{sec:preliminaries}
\subsection{Background: Haag-Kastler nets and wedge-local field in two dimensions}\label{nets}
Here we review the motivation to study wedge-local fields, the main objects which we construct
in this paper.

In the operator-algebraic approach to quantum field theory (QFT), a model of QFT
is realized as a net of algebras of observables.
A {\bf Haag-Kastler net}, or a {\bf Poincar\'e covariant net (of observables)} assigns to each
open region $O \subset \mathbb{R}^d$ a von Neumann algebra $\mathcal{A}(O)$ on a common Hilbert space $\mathcal{H}$.
If $O_1$ and $O_2$ are spacelike separated, then $\A(O_1)$ and $\A(O_2)$ should commute by Einstein causality.
In addition, one assumes that there is a continuous unitary representation $U$
of the Poincar\'e group on $\mathcal{H}$ and an invariant ground state, the vacuum $\Omega$.
The triple $(\mathcal{A},U,\Omega)$ is subject to standard axioms and considered as
a model of quantum field theory \cite{Haag96}.

If one has a Wightman field $\phi$ with certain regularity conditions, then one can construct the corresponding net by
defining ${\mathcal A}(O) := \{e^{i\phi(f)}: {\rm supp} f \subset O\}''$, where ${\mathcal M}'$ means the set
of bounded operators commuting with any element of ${\mathcal M}$. The double commutant ${\mathcal M}''$
is the smallest von Neumann algebra which includes ${\mathcal M}$. Actually it is required that
$\phi(f)$ and $\phi(g)$ have commuting spectral projections for $f, g$ with spacelike
separated support, and this follows from the regularity condition.
In this way, a Haag-Kastler net is considered as the operator-algebraic formulation of quantum field theory.

One of the difficulties in constructing Haag-Kastler nets lies in the infiniteness of the family
$\{{\mathcal A}(O)\}$ which needs to comply with the axioms. Instead, Borchers observed that for $d=2$,
actually the whole net can be recovered from the single von Neumann algebra ${\mathcal A}(W_{\mathrm R})$
associated with the (right-)wedge-shaped regions $W_{\mathrm R} := \{a \in {\mathbb R}^2: a_1 > |a_0|\}$ and the
spacetime symmetry $U$ (under a condition called Haag-duality). Furthermore, by the
Tomita-Takesaki theory of von Neumann algebras \cite{TakesakiII}, it is enough to know
the restriction of $U$ to the translation subgroup ${\mathbb R}^2$ \cite{Borchers92}.

A {\bf Borchers triple} $({\mathcal M},T,\Omega)$ consists of a von Neumann algebra ${\mathcal M}$ on ${\mathcal H}$, a unitary
representation $T$ of ${\mathbb R}^2$ with joint spectrum in the closed positive
lightcone $V_+$ and a vacuum vector $\Omega$
such that $\Omega$ is invariant under $T(a)$, ${{\rm Ad\,}} T(a) {\mathcal M} \subset {\mathcal M}$ for $a \in W_{\mathrm R}$
and ${\mathcal M}\Omega$ and ${\mathcal M}'\Omega$ are dense in ${\mathcal H}$ (these properties are called
{\bf cyclicity} and {\bf separating property} of $\Omega$ for ${\mathcal M}$, respectively).
It is easy to see that if $({\mathcal A},U,\Omega)$ is a Poincar\'e covariant net,
then $({\mathcal A}(W_{\mathrm R}), U|_{{\mathbb R}^2}, \Omega)$ is a Borchers triple.

Conversely, starting with a Borchers triple $({\mathcal M},T,\Omega)$, one can define a net
as follows: in two-spacetime dimensions, any double cone can be represented as the intersection
of two-wedges $(W_{\mathrm R}+a)\cap (W_{\mathrm L}+b) =: D_{a,b}$, where $W_{\mathrm L}$ is the reflected (left-)wedge.
Then one defines first von Neumann algebras ${\mathcal A}(D_{a,b})$ for double cones $D_{a,b}$
by ${\mathcal A}(D_{a,b}) := {{\rm Ad\,}} T(a)({\mathcal M}) \cap {{\rm Ad\,}} T(b)({\mathcal M}')$. For a general region $O$ one takes
${\mathcal A}(O) := \left(\bigcup_{D_{a,b}\subset O} {\mathcal A}(D_{a,b})\right)''$. Then one can show that
this ``net'' ${\mathcal A}$, a collection of algebras, satisfies isotony and locality. Furthermore, the representation $T$
extends to a representation $U$ of the Poincar\'e group which makes ${\mathcal A}$ covariant and
$\Omega$ is still invariant. In this way one obtains a ``net'' $({\mathcal A}, U, \Omega)$,
where the only missing property is that $\Omega$ is cyclic for ${\mathcal A}(O)$.

Hence, in the operator-algebraic approach, the construction of Haag-Kastler nets can be
split into two steps: (1) to construct Borchers triples, (2) to prove the cyclicity of $\Omega$.
In the following, we exhibit an attempt to (1).
For this purpose, we construct wedge-local fields.
Wedge-local fields $(\phi, \phi')$ is a pair of operator-valued distributions
such that $[\phi(f),\phi'(g)] = 0$ if $\supp f \subset W_\L$ and $\supp g \subset W_\R$.
Then it is natural to expect that $\M =\{e^{i\phi'(f)}:\supp g \subset W_\R\}''$,
together with appropriate $U$ and $\Omega$, gives a Borchers triple
(actually, this last step is our open problem).

\subsection{Zamolodchikov-Faddeev algebra and wedge-local fields}\label{zf}

We consider quantum field theories in $1+1$ dimensional Minkowski space (and with the convention $x \cdot y = x_0 y_0 - x_1 y_1$)
with factorizing S-matrices, which are characterized by the particle
spectrum and the two-particle scattering function as the main input in the theory. In the following,
we introduce the mathematical framework and the notation we will use to describe these models,
following \cite{Lechner03, Lechner08}. In particular, this subsection is
meant to be an overview on previous work on the topic \cite{Lechner03}
and to be an introduction to Section \ref{sec:scalar}, where we will consider models
with scalar two-particle S-matrices which have poles in the physical strip
$\RR + i(0,\pi)$.

\subsubsection{Single-particle space, S-symmetric Fock space, space-time
symmetries}\label{subsubfock}

We parametrize the momentum of a single particle with mass $m>0$
by the rapidity $\theta$:
\begin{equation*}
p(\theta):= m \left( \begin{array}{ccc}
\cosh\theta \\
\sinh\theta
\end{array} \right), \quad \theta \in \mathbb{R}.
\end{equation*}
The two-particle scattering function is generally a complex-valued meromorphic function
$S: \RR + i(0,\pi) \rightarrow \mathbb{C}$ with a certain number of symmetry properties.

These symmetry properties are the well-known properties of unitarity, hermitian analyticity
and crossing symmetry (see for example \cite{Korff:2000}) typically fulfilled by any
two-particle scattering function in a local integrable quantum field theory.
We summarize these properties in Section \ref{scalar-s}.

Using this function $S(\theta)$,
one can define a representation $D_n$ of the permutation group $\mathfrak{G}_n$
on $L^2(\mathbb{R}^n) = L^2(\RR)^{\otimes n}$, which acts as
\begin{equation*}
(D_n(\sigma)\Psi_n)(\pmb{\theta}) = S^\sigma(\pmb{\theta})\Psi_n(\pmb{\theta}^\sigma), \quad \sigma \in \mathfrak{G}_n,
\end{equation*}
where $\pmb{\theta}:=(\theta_1,\ldots,\theta_n)$, $\pmb{\theta}^\sigma:=(\theta_{\sigma(1)},\ldots,\theta_{\sigma(n)})$ and the factors $S^{\sigma}(\pmb{\theta})$ are given by
\begin{equation*}
S^\sigma(\pmb{\theta}) := \prod_{\substack{j<k\\ \sigma(j)>\sigma(k)}}S(\theta_{\sigma(j)}-\theta_{\sigma(k)}).
\end{equation*}
We are in particular interested in the space of {\bf S-symmetric functions} in $L^2(\mathbb{R}^n)$,
namely functions which are invariant under this action of $\mathfrak{G}_n$:
for any permutation $\sigma \in \mathfrak{G}_n$ it holds that
\begin{equation*}
\Psi_n(\pmb{\theta}) = S^{\sigma}(\pmb{\theta})\Psi_n(\pmb{\theta}^\sigma).
\end{equation*}
With these functions we can define the Hilbert space $\mathcal{H}$ of the theory.
In the case of models with only one species of massive scalar particle,
the single particle Hilbert space is $\mathcal{H}_1 = L^2(\mathbb{R},d\theta)$.
We define the $n$-particle Hilbert space $\mathcal{H}_n$ as the subspace of
$S$-symmetric functions in $L^2(\mathbb{R}^n) = \H_1^{\otimes n}$ and the Hilbert space
of the theory as $\mathcal{H} := \oplus_{n=0}^\infty \mathcal{H}_n$ with
$\mathcal{H}_0 =\mathbb{C}\Omega$.
We introduce the orthogonal projection $P_n := \frac{1}{n!}\sum_{\sigma \in \mathfrak{G}_n} D_n(\sigma)$
thus we can write $\mathcal{H}_n = P_n \mathcal{H}_1 ^{\otimes n}$, and we denote with $\mathcal{D}$
the dense subspace of $\mathcal{H}$ of vectors with finite particle number.

Furthermore, the Hilbert space $\mathcal{H}$ is endowed with a unitary representation of the proper
orthochronous Poincar\'e group, denoted by $\mathcal{P}_+^{\uparrow}$, which in two-dimensions consists
of translations and Lorentz boosts, acting on $\Psi =\oplus_{n=0}^\infty \Psi_n \in \mathcal{H}$ as follows
\begin{equation*}
(U(x,\lambda)\Psi)_n(\pmb{\theta})
 :=e^{i\sum_{k=1}^n p(\theta_k)\cdot x}\Psi_n(\theta_1 - \lambda, \cdots, \theta_n - \lambda).
\end{equation*}
The space-time reflection acts on the Hilbert space by an antiunitary representation $U(j)=:J$ as
\begin{equation*}
(U(j)\Psi)_n(\pmb{\theta}) := \overline{\Psi_n(\theta_n, \ldots, \theta_1)}.
\end{equation*}
Note that this action of the (proper) Poincar\'e group preserves the $S$-symmetrized
spaces $\H_n$.

We will adopt the following convention for the one-particle wave function associated with
$g \in \mathscr{S}(\mathbb{R}^2)$ \cite{Lechner03}:
\begin{equation*}
g^{\pm}(\theta) := \frac{1}{2\pi} \int d^2 x\, g(x) e^{\pm ip(\theta)\cdot x}.
\end{equation*}
If $g$ is supported in $W_\R$, then $g^+(\theta)$ has an entire analytic continuation,
which is bounded in $\RR +i(-\pi, 0)$. Furthermore, let $g_j (x) := \overline{g(-x)}$. Then, if $g$ is real, $(g_j)^+(\theta) = g^-(\theta)$
and $g^+(\overline \zeta) = \overline{g^-(\zeta)}$.
Finally, the proper Poincar\'e group acts on $\RR^2$ and also on the space of test functions naturally,
which we denote by $g_{(a,\lambda)}$ (the space-time reflection acts by $g \mapsto g_j$).
One can easily check that it is compatible with the action on the one-particle space:
\[
 (g_{(x, \lambda)})^+(\theta) = U_1(a,\lambda)g^+(\theta).
\]

\subsubsection{Zamolodchikov-Faddeev algebra, generalized creation and annihilation operators}\label{subsub:zamo}

Recalling \cite{Lechner08,Lechner03}, we consider a representation of the Zamolodchikov-Faddeev algebra
in terms of certain generalized creation and annihilation operators acting on $\mathcal{H}$,
$z^\dagger(\theta), z(\theta)$. These operator-valued distributions can be defined by their action on
$\Psi = (\Psi_n) \in \D$ as follows:
\begin{align*}
(z^\dagger(\theta)\Psi)_{n+1}(\pmb{\lambda}) &= \frac{\sqrt{n+1}}{(n+1)!}\sum_{\sigma \in \mathfrak{G}_{n+1}}S^\sigma(\pmb{\lambda})\delta(\theta -\lambda_{\sigma(1)})\Psi_n(\lambda_{\sigma(2)}, \ldots, \lambda_{\sigma(n+1)}),\\
(z(\theta)\Psi)_{n-1}(\pmb{\lambda}) &= \sqrt{n}\Psi_n(\theta,\pmb{\lambda}),
\end{align*}
and they formally fulfill the following algebraic relations:
\begin{align*}
z^\dagger(\theta)z^\dagger(\theta') &= S(\theta -\theta')z^\dagger(\theta')z^\dagger(\theta),\\
z(\theta)z(\theta') &= S(\theta -\theta')z(\theta')z(\theta),\\
z(\theta)z^\dagger(\theta') &= S(\theta' -\theta)z^\dagger(\theta')z(\theta)+\delta(\theta -\theta')\1.
\end{align*}
We shall note that $z^\dagger(f) = \int d\theta\, f(\theta) z^\dagger(\theta)$ are unbounded operators
on the space $\D$ of finite particle number states.

They can alternatively be defined in terms of the corresponding unsymmetrized creators and
annihilators $a(f), a^\dagger(f)$ (acting from the left), $f \in \mathcal{H}_1$, by setting $z^{\#}(f):=Pa^{\#}(f)P$,
where $P:= \bigoplus_{n=0}^\infty P_n$ is the orthogonal projection from the unsymmetrized Fock
space to the $S$-symmetric Fock space $\mathcal{H}$ \cite{Lechner03}.

\subsubsection*{Borchers triples for analytic S-matrices}
For the class of two-particle scattering functions $S(\theta)$ which are \emph{analytic in the physical strip} $\theta \in \mathbb{R} +i(0,\pi)$, local observables associated with wedge-regions, say with the standard left wedge $W_{\mathrm L}$, can be constructed by following an argument due to Schroer \cite{Schroer97} and Lechner \cite{Lechner03}. Specifically, they define a quantum field $\phi$ as
\begin{equation*}
\phi(f):= z^\dagger(f^+)+ z(J_1f^-), \quad f \in \mathscr{S}(\mathbb{R}^2).
\end{equation*}
We note that this reduces to the free field if $S(\theta)=1$.
In general, the field $\phi$ shares many properties with the free field as shown in
\cite[Proposition 4.2.2]{Lechner:2006}. In particular, it is defined on the subspace $\D$
of $\mathcal{H}$ of vectors with finite particle number and it is essentially
self-adjoint on $\D$ for real-valued $f$ (we denote its closure by the same symbol $\phi(f)$). It has the Reeh-Schlieder property, it solves
the free Klein-Gordon equation
and it transforms covariantly under the representation
$U(x,\lambda)$ of the proper orthochronous Poincar\'e group.
The only exception is the property of locality. The field $\phi(x)$ is not localized
at the space-time point $x$ in the usual sense, but rather in an infinitely extended wedge
with tip at $x$, $W_{\mathrm L} +x$. To make this more precise,
we introduce the ``reflected'' Zamolodchikov-Faddeev operators,
\begin{equation*}
z(\theta)' := J z(\theta)J, \quad z^\dagger(\theta)' := J z^\dagger(\theta)J,
\end{equation*}
and we define a new field $\phi'$ as, $f \in \mathscr{S}(\mathbb{R}^2)$,
\begin{equation*}
\phi'(f):= J \phi(f_j)J.
\end{equation*}
It has been shown in \cite[Proposition 2]{Lechner03} that the two fields $\phi,\phi'$
are relatively {\bf wedge-local}, in the sense  that the commutator $[e^{i\phi(f)}, e^{i\phi'(g)}]$
is zero for any real-valued test functions $f,g$ with $\supp f \subset W_\L$ and $\supp g \subset W_\R$,
Hence, we can interpret $\phi, \phi'$ as observables measurable in the wedges $W_\L, W_\R$, respectively.
This result can be obtained by computing the commutators of $z^\#$ with $z'^\#$
as shown in \cite[Lemma 4.2.5]{Lechner:2006} and by shifting a certain integral contour
which critically uses the analyticity of the two-particle scattering function $S(\theta)$
in the physical strip $\theta \in \mathbb{R} +i(0,\pi)$.

It should be remarked \cite[Proposition 2]{Lechner03} that also the properties of
the test functions $f,g$ play an important role in the proof of wedge-locality.
More specifically, the proof uses the fact that if $f \in \mathscr{S}(W_\L)$
(similar arguments apply to $g$ as well) then its Fourier transform $f^+$
fulfills certain analyticity, boundedness and symmetry properties in the strip $\mathbb{R} +i(0,\pi)$.

Starting from the fields $\phi, \phi'$ one can then define the corresponding von Neumann algebras.
The {\bf right-wedge algebra} is given by
\begin{equation*}
\M =\{e^{i\phi'(g)}: g \in \mathscr{S}(\mathbb{R}^2) \text{ real-valued},\, \supp g \subset W_\R\}'',
\end{equation*}
and other wedge algebras are defined by using the action of translations and reflection,
\begin{equation*}
\mathcal{A}(W_\R +x) := U(x,0)\M U(x,0)^*, \quad \mathcal{A}(W_\L +y):=J \mathcal{A}(W_\R -y)J.
\end{equation*}
The most important consequence of commutativity between $\phi$ and $\phi'$ is that
$\Omega$ is separating for $\M$.
The strictly local observables can then be recovered from the intersection of a right
and left wedge algebras, as explained in Section \ref{nets}.
Furthermore, Lechner proved \cite{Lechner08} that the algebras $\mathcal{A}(O)$ are nontrivial,
namely, are different from $\CC \1$.
In order to show this, he proved {\bf modular nuclearity condition},
a spectral property of the modular operators associated with the theory.
He showed that under a certain regularity condition on the scattering
function $S$ for all regions $O$, at least with a minimal size \cite{Alazzawi14},
the vacuum state $\Omega$ is cyclic and separating for the algebra $\mathcal{A}(O)$.

For the class of two-particle scattering functions $S(\theta)$ which are
\emph{not analytic in the physical strip} $\theta \in \mathbb{R} +i(0,\pi)$,
the fields $\phi(f), \phi'(f)$ fail to be wedge-local and the situation becomes
more complicated, as we will see in Section \ref{chi}.

\subsection{Poles in the S-matrix and bound states}\label{boundstates}
We now suppose that the two-particle scattering function $S(\zeta)$ has poles in the
{\bf physical strip} $\zeta \in \mathbb{R} +i(0,\pi)$. Physically,
these poles are related to the notion of ``bound state'', which here is interpreted
as the ``fusion'' of two bosons, and, although the S-matrix is factorizing,
their components are analytically related to each other.
Let us recall the physicists' arguments (see, for example, \cite{Dorey97}).
This section only tries to physically motivate the conditions on the $S$-matrices in Section \ref{scalar-s}.

Two elementary particles can be fused into a bound state if the total momentum of the two particles,
say ``1'' and ``2'', lies on the mass shell of a third particle, say ``b'' \cite{Quella99}, namely
\begin{equation*}
(p_{m_1}(\zeta_1) +p_{m_2}(\zeta_2))^2 = m_b^2,
\end{equation*}
where $m_b$ is the mass of the third particle and $p_{m_i}(\theta_i) := m_i (\cosh\theta_i, \sinh\theta_i)$.

Therefore, one can find $\zeta_b$ such that the momenta of the particles are related by
\begin{equation}\label{eqbound}
p_{m_1}(\zeta_1) + p_{m_2}(\zeta_2) = p_{m_b}(\zeta_b),
\end{equation}
where $\zeta_1, \zeta_2$ and $\zeta_b$ are the (possibly complex) rapidities of the two fusing bosons and of the bound particle, respectively.

To determine the rapidities of the particles involved and the position of the pole in the rapidity complex plane, we essentially need to solve Equation \eqref{eqbound}. In preparation for Section \ref{sec:scalar}, we will do that in a simpler case by considering a system with only one species of particle. In that case, the fusion process becomes quite simple: Two bosons of the same species fuse to form another boson of the same species, meaning also that the masses of the particles are equal.

To solve Equation \eqref{eqbound}, we make the \emph{ansatz} that the difference of
the rapidities of the fusing bosons is purely imaginary, that is,
\begin{equation*}
\zeta_1 -\zeta_2 = i\lambda.
\end{equation*}
Hence, we can parametrize the rapidities of the two fusing particles and of
the bound particle as $\zeta_1 = \theta +i\lambda_1$, $\zeta_2 = \theta +i\lambda_2$
and $\zeta_b =\theta$, with $\theta$ real. Using this parametrization for the momenta
of the particles, one can show that there is a unique solution (up to addition of $2\pi i$)
to Equation \eqref{eqbound} given by $\lambda_1 = \frac{\pi}{3}$ and $\lambda_2 = -\frac{\pi}{3}$.
Demanding that the difference $\lambda_1 -\lambda_2$ is in the physical strip,
we obtain $\lambda = \frac{2\pi}{3}$.

In the physical literature, one associates bound states with the poles of the S-matrix \cite{Dorey97}.
Specifically, we assume that the two-particle scattering function $S(\zeta)$ has a simple pole
at the point $i\lambda$ (the so called {\bf s-channel} pole), and we denote its residue by
$R:= \operatorname{res}_{\zeta = \frac{2\pi i}{3}}S(\zeta)$.

Due to the property of crossing symmetry (see \ref{crossing} in Section \ref{scalar-s}) of
the scattering function, $S(\zeta)$ has another pole at $\frac{\pi i}{3}$ (the so called
{\bf t-channel} pole) with residue $R' := \operatorname{res}_{\zeta = \frac{\pi i}{3}}S(\zeta)$.

Moreover, one can show that as a consequence of hermitian analyticity \ref{hermitian}
and unitarity \ref{unitarity} of the scattering function, the residue $R$ is purely imaginary
and one has $R' = -R$ again by crossing symmetry of $S$. We assume that $\Im R >0$.
Correspondingly, unitarity \ref{unitarity} also implies that $S(\zeta)$ has zeros at
the points $\zeta= -\frac{2\pi i}{3}, -\frac{\pi i}{3}$.

Finally, we will assume that except for the two simple poles at $\frac{2i\pi}{3},\frac{\pi i}{3}$
there are no other poles of $S$ in the physical strip. An example of a $1+1$-dimensional
integrable model with scalar S-matrix fulfilling these properties is the
Bullough-Dodd model \cite{FMS93}. A particle in these models is considered
as the bound state of two of the same species, and the two-particle S-matrix $S$ satisfies
a nontrivial equation, the so-called bootstrap equation \ref{bootstrap}.

\section{Scalar two-particle S-matrices}\label{sec:scalar}
Here we present our wedge-local fields associated with scalar S-matrices with poles in the physical strip.
We specify the properties of these S-matrices and give examples.
Then an important operator $\chi(f)$, which ``binds $f$ and the state'', is introduced.
We use this operator in order to construct our wedge-local fields.

\subsection{Properties of scalar S-matrix with poles}\label{scalar-s}
Let us consider the simplest case where the one-particle space has multiplicity one.
Since we are interested in S-matrices with poles in the physical strip, the bound state of two ``elementary''
particle must be the same particle, therefore must have the same mass.
As discussed in Section \ref{boundstates}, one observes that
the only possibility for the pole in s-channel is $\frac{2\pi i}3$.
By crossing symmetry, the S-matrix must have another pole at $\frac{\pi i}3$.
Furthermore, the residue of $S$ at $\frac{2\pi i}3$ must be a positive multiple of
the imaginary unit $i$. This is interpreted as a consequence of hermiticity of
the Hamiltonian \cite{CM89}\footnote{In the last paragraph of this paper \cite{CM89}, they claim that
a scalar S-matrix is impossible. The argument is incomplete, because they assume that
such an S-matrix should have the same zeros as the simplest S-matrix which does not satisfies
the condition on the residue, but actually zeros do not necessarily
have physical meaning and there is no reason to exclude them. Indeed,
one of the author of the same paper \cite{CM89} published
later a paper on the Bullough-Dodd model \cite{FMS93},
which is believed to have a scalar S-matrix and be unitary.}

Summarizing, we assume that our two-particle S-matrix $S$ is defined on $\RR + i(0, \pi)$
except for poles indicated below, has therefore $L^\infty$-boundary values at $\RR$ and $\RR + \pi i$, and
has the following properties (c.f.\! \cite{LS14}).
\begin{enumerate}
\renewcommand{\theenumi}{(S\arabic{enumi})}
 \renewcommand{\labelenumi}{\theenumi}
 \item \label{unitarity}  {\bf Unitarity.} $S(\theta)^{-1} = \overline{S(\theta)},\;\; \theta \in \RR$.
 \item \label{hermitian} {\bf Hermitian analyticity.} $S(-\theta) = S(\theta)^{-1},\;\; \theta \in \RR$.
 \item \label{crossing} {\bf Crossing symmetry.} $S$ is meromorphic and the boundary values satisfy
 $S(\theta) = S(\pi i - \theta), \theta \in \RR$.
 \item \label{bootstrap} {\bf Bootstrap equation.} $S\left(\theta + \frac{\pi i}3\right) = S(\theta)S\left(\theta + \frac{2\pi i}3\right), \theta \in \RR$.
 \item \label{poles-boundedness} {\bf Positive residue.} $S$ has a simple pole at $\frac{2\pi i}3$ and $\displaystyle{\underset{\zeta = \frac{2\pi i}3} \res S(\zeta) \in i\RR_+}$.
 Except this and the pole at $\frac{\pi i}3$, there is no pole in the physical strip $\zeta \in \RR + i(0,\pi)$ and
 $S(\zeta)$ is bounded in the complement of the union of neighborhoods of these poles. 
 \item \label{fermionic} {\bf Value at zero.} $S(0) = -1$.
\end{enumerate}
Actually, we will see in Appendix \ref{classification} that \ref{fermionic}
follows from \ref{unitarity}--\ref{poles-boundedness}.
Note also that we consider bosonic particles,
yet the observables are represented on the $S$-symmetric Fock space (see below).

We introduce $\eta = i\sqrt{2\pi|R|}$, in accordance with the literature \cite{Quella99}, recalling that $R$ is the residue of $S$ at $\zeta = \frac{2\pi i}3$.

\subsubsection*{Examples}
Although specific expressions of $S$ are not needed in our main construction,
we present here a family of examples, in order to show that the above set of axioms is not empty.

The simplest examples of $S$ which are believed to be associated to the Bullough-Dodd model \cite{FMS93}
are the following:
\[
 S_B(\theta) = \frac{\tanh\frac12\left(\theta + \frac{2\pi i}3\right)}{\tanh\frac12\left(\theta - \frac{2\pi i}3\right)} \cdot
 \frac{\tanh\frac12\left(\theta + \frac{(B-2)\pi i}3\right)}{\tanh\frac12\left(\theta - \frac{(B-2)\pi i}3\right)}
 \frac{\tanh\frac12\left(\theta - \frac{B\pi i}3\right)}{\tanh\frac12\left(\theta + \frac{B\pi i}3\right)},
\]
where $0 < B < 2, B \neq 1$. If we introduce the notation $f_A(\theta) := \frac{\tanh\frac12(\theta + A\pi i)}{\tanh\frac12(\theta - A\pi i)}$,
we can write it as $S_B(\theta) = f_\frac23(\theta)f_{\frac{B}3 - \frac23}(\theta) f_{-\frac{B}3}(\theta)$.
It holds that $S_B(\theta) = S_{2-B}(\theta)$. Furthermore, $S_1(\theta) = f_{-\frac23}(\theta)$ which has no pole
in the physical strip $\RR + i(0,\pi)$, therefore we exclude this case.

It is interesting to note that the first factor
\begin{align*}
  f_\frac23(\theta) &:= \frac{\tanh\frac12\left(\theta + \frac{2\pi i}3\right)}{\tanh\frac12\left(\theta - \frac{2\pi i}3\right)}
  = \frac{\sinh\frac12\left(\theta + \frac{2\pi i}3\right)}{\cosh\frac12\left(\theta + \frac{2\pi i}3\right)}
  \frac{\cosh\frac12\left(\theta - \frac{2\pi i}3\right)}{\sinh\frac12\left(\theta - \frac{2\pi i}3\right)} \\
  &= -\frac{\sinh\frac12\left(\theta + \frac{\pi i}3\right)}{\sinh\frac12\left(\theta - \frac{\pi i}3\right)}
  \frac{\sinh\frac12\left(\theta + \frac{2\pi i}3\right)}{\sinh\frac12\left(\theta - \frac{2\pi i}3\right)}
\end{align*}
satisfies all properties of the S-matrix but positivity of residue.
Indeed, we have
\begin{align*}
&\underset{\zeta = \frac{i\pi}3}\res f_\frac23(\zeta)=
-\frac{\sinh\frac12\left(\frac{i2\pi}3\right)}{\frac12}\cdot
\frac{\sinh\frac12(i\pi)}{\sinh\frac12\left(-\frac{i\pi}3\right)} = 2\sqrt 3 i, \\
&\underset{\zeta = \frac{i2\pi}3}\res f_\frac23(\zeta)=
-\frac{\sinh\frac12\left(i\pi\right)}{\sinh\frac12\left(\frac{i\pi}3\right)}\cdot
\frac{\sinh\frac12\left(\frac{i4\pi}3\right)}{\frac12} = -2\sqrt 3 i.
\end{align*}

The remaining factor
\[
 f_{\frac{B}3 - \frac23}(\theta) f_{-\frac{B}3}(\theta) =
 \frac{\tanh\frac12\left(\theta + \frac{(B-2)\pi i}3\right)}{\tanh\frac12\left(\theta - \frac{(B-2)\pi i}3\right)} \cdot
 \frac{\tanh\frac12\left(\theta - \frac{B\pi i}3\right)}{\tanh\frac12\left(\theta + \frac{B\pi i}3\right)},
\]
also satisfies \ref{unitarity}--\ref{bootstrap} except for positivity of residue, namely, it
has no pole in the physical strip and satisfies
\[
 f_{\frac{B}3 - \frac23}\left(\frac{2\pi i}3\right) f_{-\frac{B}3}\left(\frac{2\pi i}3\right) =
 \frac{\tan\left(\frac{B\pi}6\right)}{\tan\left(\frac{(4-B)\pi}6\right)}
 \frac{\tan\left(\frac{(2-B)\pi}6\right)}{\tan\left(\frac{(B+2)\pi}6\right)} < 0,
\]
where $0 < B < 1$ or $1 < B < 2$.
Therefore, the product $S_B(\theta) = f_\frac23(\theta)\cdot f_{\frac{B}3 - \frac23}(\theta) f_{-\frac{B}3}(\theta)$
has a positive residue at $\frac{2\pi i}3$, and therefore, satisfies \ref{poles-boundedness}.
\ref{fermionic} is now straightforward.

 From the above computations, it is also clear that a product
\[
 S_{B_1, B_2, \cdots B_n}(\theta) := f_\frac23(\theta) \prod_{k=1}^n f_{\frac{B_k}3 - \frac23}(\theta) f_{-\frac{B_k}3}(\theta),
\]
where $0< B_k < 2, B_k \neq 1$ and $n$ is odd, satisfies all the properties of S-matrix.
When $n > 1$, no Lagrangian is known for such two-particle S-matrix.
We will completely classify all S-matrices which comply with \ref{unitarity}--\ref{fermionic}
in Appendix \ref{classification}.

\subsection{The bound-state operator}\label{chi}
For a test function $f$ supported in $W_\L$, let us introduce an unbounded operator $\chi(f)$ on the $S$-symmetric Fock space $\H$.
This operator will preserve the particle number and will be interpreted as the operator which
``makes a bound state''. Actually, in a model with scalar S-matrix, a bound state particle
of two ``elementary particles'' is again the same elementary particle as there is only
one species of particle \cite{FMS93}.

Let us denote its component on $\H_n$ by $\chi_n(f)$.
Firstly, $\chi_0(f)$ annihilates the vacuum $\Om$.

For $\xi \in \H_1 = L^2(\RR,d\theta)$, we say that $\xi(\theta)$ has an {\bf $L^2$-bounded analytic continuation} on a strip parallel to $\RR$ (e.g. $\RR + i(-\epsilon,0)$ or $\RR + i(0,\epsilon), \epsilon > 0$)
if $\xi(\zeta)$ is an analytic function on that strip (with boundary value $\xi(\theta)$ at $\im \zeta =0$)
such that for each fixed $-\epsilon < \a < 0$ (respectively $0 < \a <\epsilon$) the function $\theta \mapsto \xi(\theta + i\a)$ is an $L^2$-function
in $\theta$, with uniform $L^2$-bound in $\alpha$.
The action $\chi_1(f)$ on $\H_1$ is given as follows:
\begin{equation}\label{eq:chi1}
\begin{aligned}
 \dom(\chi_1(f)) &:= \left\{\xi \in \H_1: \xi(\theta) \mbox{ has an } L^2\mbox{-bounded analytic continuation to }\theta - \frac{i\pi}3\right\}, \\
 (\chi_1(f)\xi)(\theta) &:= -i\eta f^+\left(\theta + \frac{i\pi}3\right) \xi\left(\theta - \frac{i\pi}3\right)
 = \sqrt{2\pi |R|}f^+\left(\theta + \frac{i\pi}3\right) \xi\left(\theta - \frac{i\pi}3\right),
\end{aligned}
\end{equation}
where $\eta$ and $R$ are given in Section \ref{scalar-s}. Note that
$f^+(\theta + \frac{\pi i}3)$ is bounded, therefore, $\chi_1(f)\xi$ is $L^2$.
Then, we define:
\begin{align*}
 \chi_n(f) &:= nP_n(\chi_1(f)\otimes\1\otimes\cdots\otimes\1)P_n, \\
 \chi(f) &= \bigoplus_{n=0}^\infty \chi_n(f),
\end{align*}
where $\dom(\chi(f))$ is the algebraic direct sum of $\dom(\chi_n(f))$,
hence is a subspace of $\D$ and, of course, the domain of the product $AB$ of possibly unbounded operators $A,B$ is
given by $\{\xi \in \dom(B): B\xi \in \dom(A)\}$.
We will discuss the question of self-adjoint extensions of these operators in Appendix \ref{domain}.

Let $\tau_j \in \mathfrak{S}_n$, $1 \le j \le n-1$, be the transposition which exchanges $j$ and $j+1$, and let  $\rho_k =\tau_{k-1}\cdots\tau_1$ be the cyclic permutation
\[
 \rho_k: (1,2,\cdots, n) \mapsto (k,1,2, \cdots, k-1, k+1,\cdots, n);
\]
note that $\rho_1$ is the unit element of $\mathfrak{S}_n$. With this notation, since $\sigma$ is a surjection of $\{1,\cdots, n\}$ onto itself, any permutation $\sigma \in \mathfrak{S}_n$ can be written as the product $\rho_{\sigma(1)}\underline\sigma$
with a permutation $\underline{\sigma}$ of $n-1$ numbers $(2,3,\cdots, n)$.

The operator $\chi_1(f)\otimes\1\otimes\cdots\otimes\1$ commutes with such $D_n(\underline{\sigma})$.
As $P_n = \frac1{n!}\sum_{\sigma\in\mathfrak{S}_n} D_n(\sigma)$, it holds that $D_n(\sigma)P_n = P_n$
and the $n$-particle component $\chi_n(f)$ can alternatively be expressed as follows:
\begin{align*}
 \chi_n(f) &= nP_n(\chi_1(f)\otimes\1\otimes\cdots\otimes\1)P_n \\
 &= \frac1{(n-1)!}\sum_{\sigma\in\mathfrak{S}_n} D_n(\rho_{\sigma(1)})D_n(\underline{\sigma})
 (\chi_1(f)\otimes\1\otimes\cdots\otimes\1)P_n \\
 &= \frac1{(n-1)!}\sum_{\sigma\in\mathfrak{S}_n} D_n(\rho_{\sigma(1)})
 (\chi_1(f)\otimes\1\otimes\cdots\otimes\1)P_n \\
 &= \sum_{1\le k \le n} D_n(\rho_k)(\chi_1(f)\otimes\1\otimes\cdots\otimes\1)P_n.
\end{align*}

Note that, if $\Psi_n$ is $S$-symmetric and in the domain of $\chi_1(f)\otimes\1\otimes\cdots\otimes\1$,
$\Psi_n(\theta_1,\cdots,\theta_n)$ has a meromorphic
continuation in $\theta_k$ and it holds that
\[
 \Psi_n\left(\theta_1 - \frac{\pi i}3, \theta_2, \cdots, \theta_n\right)
 = \prod_{2\le j \le k} S\left(\theta_j-\theta_1 + \frac{\pi i}3\right)
\Psi_n\left(\theta_2,\cdots,\theta_k,\theta_1 - \frac{\pi i}3,\theta_{k+1},\cdots,\theta_n\right).
\]
Therefore, each term in the last expression of $\chi_n(f)$ above can be further written as follows for $k\ge 2$
and for $\Psi_n = P_n\Psi_n$:
\begin{align*}
 &(D_n(\rho_k)(\chi_1(f)\otimes\1\otimes\cdots\otimes\1)\Psi_n)(\theta_1\cdots\theta_n) \\
 &=\; \prod_{1\le j \le k-1} S(\theta_k-\theta_j)
 ((\chi_1(f)\otimes\1\otimes\cdots\otimes\1)\Psi_n)(\theta_k,\theta_1,\cdots,\theta_{k-1},\theta_{k+1},\cdots\theta_n) \\
 &=\; -i\eta\prod_{1\le j \le k-1} S(\theta_k-\theta_j) f^+\left(\theta_k + \frac{\pi i}3\right)
 \Psi_n\left(\theta_k - \frac{\pi i}3,\theta_1,\cdots,\theta_{k-1},\theta_{k+1},\cdots\theta_n\right) \\
 &=\; -i\eta\prod_{1\le j \le k-1} S(\theta_k-\theta_j) S\left(\theta_j-\theta_k + \frac{\pi i}3\right)
 f^+\left(\theta_k + \frac{\pi i}3\right)
 \Psi_n\left(\theta_1,\cdots,\theta_k - \frac{\pi i}3,\cdots\theta_n\right) \\
 &=\; -i\eta\prod_{1\le j \le k-1} S\left(\theta_k-\theta_j + \frac{\pi i}3\right)
 f^+\left(\theta_k + \frac{\pi i}3\right)
 \Psi_n\left(\theta_1,\cdots,\theta_k - \frac{\pi i}3,\cdots\theta_n\right),
\end{align*}
where we reordered the variables in the third equality,
and we used the bootstrap equation in the last equality.
Note that, although $S$-factors have poles, these poles are cancelled by the zeros of $\Psi_n$
by the definition of the domain of $\chi_n(f)$ in Equation \eqref{eq:chi1} and the whole expression remains $L^2$
(this is the meaning that $\Psi_n \in \dom(\chi_n(f))$.
In this expression, the one-particle component $\chi_1(f)$ acts on each variable of $\Psi_n$
up to a correction of $S$ factors. As the one-particle action \eqref{eq:chi1} realizes the idea that
the state of one elementary particle $\xi$ is fused with $f^+$ into the same species of
particle, as in Section \ref{boundstates}, we might call it the ``bound state operator''.

Similarly, for a test function $g$ supported in the right wedge $W_\R$,
we introduce the reflected bound state operator $\chi'(g)$:
\begin{align*}
 \dom(\chi'_1(g)) &:= \left\{\xi \in \H_1: \xi(\theta) \mbox{ has an } L^2\mbox{-bounded analytic continuation to }\theta + \frac{i\pi}3\right\}, \\
 (\chi'_1(g)\xi)(\theta) &:= -i\eta g^+\left(\theta - \frac{i\pi}3\right) \xi\left(\theta + \frac{i\pi}3\right)
 = \sqrt{2\pi|R|}g^+\left(\theta - \frac{i\pi}3\right) \xi\left(\theta + \frac{i\pi}3\right), \\
\chi'_n(g) &:= nP_n(\1\otimes\cdots\otimes\1\otimes\chi'_1(g))P_n.
\end{align*}
For a real $g$, this operator is indeed related to $\chi$ by the CPT operator $J$:
\[
 \chi'(g) = J\chi(g_j)J.
\]
To see this, let us consider the one-particle components.
Since $J_1$ acts as the complex conjugation, it takes an analytic function in the lower strip to
an analytic function in the upper strip, therefore the domains of $\chi'(g)$ and of $J\chi(g_j)J$ coincide.
Recall that for a real $g$,
$(g_j)^+(\theta) = g^-(\theta)$ and $g^+(\overline\zeta) = \overline {g^-(\zeta)}$.
If $\xi \in \dom(\chi'_1(g))$, $(J_1\xi)(\zeta) =\overline{ \xi(\overline\zeta)}$ and we have
\begin{align*}
 (\chi'_1(g)\xi)(\theta) &= \sqrt{2\pi|R|}g^+\left(\theta - \frac{\pi i}3\right)\xi\left(\theta + \frac{\pi i}3\right) \\
 &= \overline{\sqrt{2\pi|R|}g^-\left(\theta + \frac{\pi i}3\right)\overline{\xi\left(\theta + \frac{\pi i}3\right)}} \\
 &= \overline{\sqrt{2\pi|R|}g^-\left(\theta + \frac{\pi i}3\right)(J_1\xi)\left(\theta - \frac{\pi i}3\right)} \\
 &= (J_1\chi(g_j)J_1\xi)(\theta).
\end{align*}
As $J_n$ commutes with $P_n$, we have $\chi'_n(g) = J_n\chi_n(g_j)J_n$.
Since the whole operators $\chi(g)$ and $\chi'(g)$ are defined as the direct sum,
the desired equality follows.

This operator $\chi'(g)$ has an alternative expression as $\chi(f)$ does:
\begin{align*}
 \chi'_n(g) &= \sum_{1\le k \le n} D_n(\rho'_k)(\1\otimes\cdots\otimes\1\otimes\chi'_1(g))P_n,
\end{align*}
where $\rho'_k = \tau_{n-k+1}\tau_{n-k+2}\cdots\tau_{n-1}$ are the cyclic permutations
\[
 \rho'_k:(1,\cdots, n-1,n) \longmapsto (1,\cdots n-k, n-k+2,\cdots, n-1, n, n-k+1)  
\]
and
\begin{align*}
 &(D_n(\rho'_k)(\1\otimes\cdots\otimes\1\otimes\chi'_1(g))\Psi_n)(\theta_1\cdots\theta_n) \\
 =&\; \sqrt{2\pi|R|}\prod_{n-k+2 \le j \le n} S\left(\theta_j- \theta_{n-k+1} + \frac{\pi i}3\right)
 g^+\left(\theta_{n-k+1} - \frac{\pi i}3\right) \\
 &\times \Psi_n\left(\theta_1,\cdots,\theta_{n-k+1} + \frac{\pi i}3,\cdots\theta_n\right),
\end{align*}
for $k \ge 2$.

Let us check basic properties of $\chi(f)$.
As $\chi'(g)$ is defined similar to $\chi(f)$, they share a number of properties.
The following propositions have the obvious counterparts for $\chi'(g)$.

\begin{proposition}\label{pr:chi:symmetry}
 For a real test function $f$ supported in $W_\L$,
 the operator $\chi(f)$ is densely defined and symmetric.
\end{proposition}
\begin{proof}
 First let us look at $\chi_1(f)$. By its definition, it is densely defined.
 To see that it is symmetric, we take two vectors $\xi, \eta \in \dom(\chi_1(f))$ 
 which have compact Fourier transforms. Such vectors form a core for $\chi_1(f)$.
 Indeed, we can write $\chi_1(f) = \sqrt{2\pi|R|}x_f\Delta^{\frac16}_1$,
 where $x_f$ is the multiplication operator by $f^+\left(\theta + \frac{\pi i}3\right)$ and
 $\Delta_1 \xi(\theta) = \xi(\theta - 2\pi i)$ with the obvious domain.
 As $\Delta_1^{it}$ implements the real shift $(\Delta_1^{it}\xi)(\theta) = \xi(\theta + 2\pi t)$,
 the Fourier transform diagonalizes $\Delta_1$, and the vectors above form a core.
 Then, $\xi, \eta$ are the Fourier transforms of compactly supported functions,
 therefore they are rapidly decreasing. Recall that $f^+(\theta + i\alpha)$ is rapidly decreasing
 as well, if $0 \le \alpha \le \pi$.
 Furthermore, $\overline{\eta(\overline\zeta)}$ is analytic in $\zeta$.
 Now we compute:
 \begin{align*}
  \<\eta, \chi_1(f)\xi\> &= \sqrt{2\pi|R|}\int d\theta\,\overline{\eta(\theta)}f^+\left(\theta + \frac{\pi i}3\right)\xi\left(\theta - \frac{\pi i}3\right) \\
  &= \sqrt{2\pi|R|}\int d\theta\,\overline{\eta\left(\theta - \frac{\pi i}3\right)}f^+\left(\theta + \frac{2\pi i}3\right)\xi(\theta) \\
  &= \sqrt{2\pi|R|}\int d\theta\,\overline{\eta\left(\theta - \frac{\pi i}3\right)f^+\left(\theta + \frac{\pi i}3\right)}\xi(\theta) \\
  &= \<\chi_1(f)\eta, \xi\>,
 \end{align*}
 where we used the Cauchy theorem in the second equality and
 $f^+(\zeta) = \overline{f^-(\overline{\zeta})} = \overline{f^+(\overline{\zeta} + i\pi)}$ in the third.
 As this equality holds on a core of $\chi_1(f)$, so does it on the whole domain and
 we obtain the symmetry of $\chi_1(f)$.
 
 Let us check that $\chi_n(f)$ is densely defined. As $P_n$ is bounded,
 it is enough to see that $(\chi_1(f)\otimes\1\otimes\cdots\otimes \1)P_n$ is densely defined.
 The range of $P_n$ is the set of $S$-symmetric functions, while the domain of
 $\chi_1(f)\otimes\1\otimes\cdots\otimes \1$ is the functions which has an $L^2$-bounded analytic continuation
 to $\frac{\pi i}3$ in the first variable.
 Let us take an arbitrary set $\{\xi_1, \cdots \xi_n\}$ of $n$ vectors in the domain of $\chi_1(f)$.
 Then,
 \begin{align*}
  &(P_n(\xi_1\otimes \cdots \otimes \xi_n))(\theta_1,\cdots,\theta_n) 
  &= \frac{1}{n!}\left(\sum_{\sigma\in\mathfrak{S}n} \prod_{\substack{j<k\\ \sigma(j)>\sigma(k)}}S(\theta_{\sigma(j)}-\theta_{\sigma(k)})
  \cdot \xi_{\sigma(1)}(\theta_1)\cdots\xi_{\sigma(n)}(\theta_n)\right)
 \end{align*}
 is of course $S$-symmetric, but
 has poles which come from the $S$-factors. Such vectors are a dense subspace of
 $P_n\H_n$. Note that all these poles come from the poles of $S$ at $\theta_j - \theta_k = \frac{\pi li}3$,
 $l = 1,2$ and $1\le  k < j \le n$. In order to compensate these poles, we can multiply it $(n-2)!$ times by
\[
  C_n(\pmb{\theta}) := \prod_{1 \le k < j \le n} \frac{(\theta_j - \theta_k - \frac{\pi i}3)
  (\theta_k - \theta_j - \frac{\pi i}3)}{
  (\theta_j - \theta_k - i\alpha)
  (\theta_k - \theta_j - i\alpha)},
\]
 where $\alpha < 0$ or $\pi < \alpha$.
 This is a bounded invertible symmetric function in the real variables $\pmb{\theta}$.
 Therefore, the multiplication operator $M_{C_n}$ by $C_n$ preserves $\H_n$ and
 its image of a dense subspace is again dense.
 The functions in the image have an analytic continuations in the first
 (actually any) variable to $\frac{\pi i}3$ and are $S$-symmetric, therefore in the domain of
 $(\chi_1(f)\otimes\1\otimes\cdots\otimes\1)P_n$.

 Now the symmetry of $\chi_n(f)$ follows from the symmetry of $\chi_1(f)$ and
 a general fact $(ABC)^* \supset (C^*B^*A^*)$. The whole operator $\chi(f)$ is by definition the direct
 sum of these operators, therefore densely defined and symmetric as well.
\end{proof}

In Section \ref{zf}, we saw that there is an action $U$ of the Poincar\'e group on $\H$.
Let us check that $\chi(f)$ is covariant with respect to $U$.

\begin{proposition}\label{pr:chi:covariance}
 Let $f$ be a test function supported in $W_\L$ and $(a,\lambda) \in \poincare$ such that
 $a \in W_\L$. Then it holds that
 $\Ad U(a,\lambda)(\chi(f)) \subset \chi(f_{(a,\lambda)})$.
\end{proposition}
\begin{proof}
 As both $\chi(f)$ and $U$ are the direct sums $\chi(f)=\bigoplus_n \chi_n(f)$ and $U(a,\lambda) = \bigoplus_n U_n(a,\lambda)$ respectively,
 we can restrict ourselves to $\H_n$. Furthermore, by definition
 $\chi_n(f) = nP_n(\chi_1(f)\otimes\1\otimes\cdots\otimes\1)P_n$ and $P_n$ commutes
 with $U_n(a,\lambda) = U_1(a,\lambda)\otimes\cdots\otimes U_1(a,\lambda)$,
 therefore it is enough to show that $U_1(a,\lambda)\chi_1(f)U_1(a,\lambda)^* \subset \chi_1(f_{(a,\lambda)})$.
 
 Consider first a pure boost $(0,\lambda)$. Note that $U(0,\lambda)$ preserves the domain $\dom(\chi_1(f))$.
 We saw in Section \ref{zf} that $(U_1(0,\lambda)f^+)(\theta) = f^+(\theta-\lambda) = (f_{(0,\lambda)})^+(\theta)$.
 On the other hand, for $\xi\in\dom(\chi_1(f)) = \dom(\chi_1(f_{(0,\lambda)}))$ we have
 \begin{align*}
 (U_1(0,\lambda)\chi_1(f)U_1(0,\lambda)^*\xi)(\theta)
 &= (\chi_1(f)U_1(0,\lambda)^*\xi)(\theta-\lambda) \\
 &= -i\eta f^+\left(\theta - \lambda + \frac{\pi i}3\right)\cdot
 (U_1(0,\lambda)^*\xi)\left(\theta-\lambda - \frac{\pi i}3\right) \\
 &= -i\eta f^+\left(\theta - \lambda + \frac{\pi i}3\right)\cdot
 \xi\left(\theta - \frac{\pi i}3\right) \\
 &= (\chi_1(f_{(0,\lambda)})\xi)(\theta).
 \end{align*}
 Hence the covariance with respect to boosts holds.
 
 Next we take a pure translation $(a,0)$, where $a \in W_\L$.
 As $U_1(a,\lambda)^*$ multiplies by $e^{-ia\cdot p(\theta)}$, which has a bounded analytic continuation
 in $\RR + i(-\pi,0)$, the domain of $\chi_1(f)$ is preserved.
 Again, recall the one-particle action $(U_1(a,0)f^+)(\theta) = e^{ia\cdot p(\theta)}f^+(\theta) = (f_{(a,0)})^+(\theta)$.
 For a vector $\xi\in\dom(\chi_1(f)) = \dom(\chi_1(f_{(a,0)}))$, it holds that
 \begin{align*}
 (U_1(a,0)\chi_1(f)U_1(a,0)^*\xi)(\theta)
 &= e^{ia\cdot p(\theta)}(\chi_1(f)U_1(a,0)^*\xi)(\theta) \\
 &= -i\eta e^{ia\cdot p(\theta)}f^+\left(\theta + \frac{\pi i}3\right)\cdot
 (U_1(a,0)^*\xi)\left(\theta - \frac{\pi i}3\right) \\
 &= -i\eta e^{ia\cdot p(\theta)} f^+\left(\theta + \frac{\pi i}3\right)\cdot
 e^{-ia\cdot p\left(\theta - \frac{\pi i}3\right)}\xi\left(\theta - \frac{\pi i}3\right).
 \end{align*}
 As we have
 \begin{align*}
  p\left(\theta + \frac{\pi i}3\right) + p\left(\theta - \frac{\pi i}3\right)
  &= m\left(\begin{array}{c}
           \frac12 \cosh\theta + i\frac{\sqrt 3}2\sinh\theta \\
           \frac12 \sinh\theta + i\frac{\sqrt 3}2\cosh\theta 
          \end{array}\right) +
    m\left(\begin{array}{c}
           \frac12 \cosh\theta - i\frac{\sqrt 3}2\sinh\theta \\
           \frac12 \sinh\theta - i\frac{\sqrt 3}2\cosh\theta 
          \end{array}\right) \\
 &= m\left(\begin{array}{c}
           \cosh\theta  \\
           \sinh\theta 
          \end{array}\right) = p(\theta),
 \end{align*}
 or equivalently, $p(\theta) - p\left(\theta - \frac{\pi i}3\right) =  p\left(\theta + \frac{\pi i}3\right)$,
 we obtain
 \begin{align*}
 (U_1(a,0)\chi_1(f)U_1(a,0)^*\xi)(\theta)
 &= -i\eta e^{ia\cdot \left(p(\theta)- p\left(\theta - \frac{\pi i}3\right)\right)} f^+\left(\theta + \frac{\pi i}3\right)\cdot
 \xi\left(\theta - \frac{\pi i}3\right) \\
 &= -i\eta e^{ia\cdot p\left(\theta + \frac{\pi i}3\right)} f^+\left(\theta + \frac{\pi i}3\right)\cdot
 \xi\left(\theta - \frac{\pi i}3\right) \\
 &= (\chi_1(f_{(a,0)})\xi)(\theta),
 \end{align*}
 and this is the covariance with respect to translations.
\end{proof}

\subsubsection*{Formal expression}
We argue that our operator $\chi(f)$ defined above is formally equivalent to
\[
-i\eta \int d\theta\, f^+\left(\theta +\frac{\pi i}3\right) z^\dagger(\theta)z\left(\theta - \frac{\pi i}{3}\right).
\]
This formal expression preserves the particle number, because it consists of a creation operator $z^\dagger$
and an annihilation operator $z$. In this argument, we do not pay much attention to domains.
We will not use this formal expression later in proofs, yet it is interesting to
observe that the operator $\chi(f)$ has such a simple expression.

We know that $z(f) = Pa(f)P$ and $z^\dagger(f) = Pa^\dagger(f)P$,
where $P = \bigoplus_n P_n$ is the projection from the unsymmetrized Fock space to
the $S$-symmetric Fock space.
Let us take an $S$-symmetric vector $\Psi_n \in \H_n$.
The unsymmetrized annihilation operator $a(\zeta)$ substitutes the first variable by
a fixed number $\zeta$.
As $\Psi_n$ is an $S$-symmetric function with $n$ variables, then after the action of 
$a$, it is still an $S$-symmetric function with $n-1$ variables.
This means that we can remove several factors of $P$ in the above formal expression and obtain
\begin{align*}
&  -i\eta \int d\theta\, f^+\left(\theta +\frac{\pi i}3\right) z^\dagger(\theta)z\left(\theta - \frac{\pi i}{3}\right)\Psi_n \\
=&\; -i\eta P_n \int d\theta\, f^+\left(\theta +\frac{\pi i}3\right) a^\dagger(\theta)a\left(\theta - \frac{\pi i}3\right)\Psi_n.
\end{align*}
Let us look at the integrand formally. For a fixed $\theta$, the action of the annihilation operator gives
\[
\left(a\left(\theta - \frac{\pi i}3\right)\Psi_n\right)(\theta_1,\cdots, \theta_{n-1})
= \sqrt n \Psi_n\left(\theta - \frac{\pi i}3,\theta_1, \cdots, \theta_{n-1}\right).
\]
Thereafter, the action of the creation operator gives the following:
\[
\left(a^\dagger(\theta)a\left(\theta - \frac{\pi i}3\right)\Psi_n\right)(\theta_1,\cdots, \theta_n)
= n\delta(\theta-\theta_1)\Psi_n\left(\theta - \frac{\pi i}3,\theta_2, \cdots, \theta_n\right)
\]
and after the integration and the multiplication by $-i\eta$, we get
\begin{align*}
&-i\eta n\int d\theta\, f^+\left(\theta + \frac{\pi i}3\right) a^\dagger (\theta)
a\left(\theta - \frac{\pi i}3\right)\Psi_n(\theta_1,\cdots, \theta_n) \\
=&\;-i\eta n\int d\theta\, f^+\left(\theta + \frac{\pi i}3\right) \delta(\theta-\theta_1)
\Psi_n\left(\theta - \frac{\pi i}3,\theta_2, \cdots, \theta_n\right) \\
=&\; -i\eta nf^+\left(\theta_1 + \frac{\pi i}3\right)\Psi_n\left(\theta_1 - \frac{\pi i}3,\theta_2, \cdots, \theta_n\right).
\end{align*}
If we look at the action on the first variable, this is exactly $n\chi_1(f)$.
Recall that $\chi_n(f) = nP_n(\chi_1(f)\otimes\cdots\otimes\1)P_n$.
As $\Psi_n$ was arbitrary, we get finally
\[
 \chi_n(f)\Psi_n = -i\eta \int d\theta\, f^+\left(\theta +\frac{\pi i}3\right) z^\dagger(\theta)z\left(\theta - \frac{\pi i}{3}\right)\Psi_n,
\]
which is the desired expression.
Similarly, we have
\[
 \chi'_n(g) = -i\eta \int d\theta\, g^+\left(\theta -\frac{\pi i}3\right) z'^\dagger(\theta)z'\left(\theta + \frac{\pi i}{3}\right).
\]

\subsection{The wedge-local fields}\label{field-scalar}
We define the (left) wedge-local field by
\[
 \fct(f) = \phi(f) + \chi(f).
\]
The domain of $\phi(f)$ includes that of $\chi(f)$, therefore the domain of $\fct(f)$
coincides with the latter.

The reflected field is given by
\[
 \fct'(g) := \phi'(g) + \chi'(g) = J\fct(g_j)J.
\]

\begin{proposition}
 Let $f$ be a real test function supported in $W_\L$. Then the operator $\fct(f)$ defined
 above has the following properties.
 \begin{enumerate}[{(}1{)}]
  \item $\fct(f)$ is a symmetric operator. \label{pr:fct:symmetry}
  \item $\fct$ is a solution of the Klein-Gordon equation, in the sense that
  $\fct((\Box + m^2)f) = 0$. \label{pr:fct:kg}
  \item $\fct(f)$ is covariant with respect to $U$, in the sense that if $f$ is supported in $W_\L$
  and if $(a,\lambda) \in \poincare$ such that $a \in W_\L$, then
  it holds that $U(g)\fct(f)U(g)^* \subset \fct(f_{(a,\lambda)})$. \label{pr:fct:covariance}
 \end{enumerate}
\end{proposition}
\begin{proof}
 (\ref{pr:fct:symmetry}) We saw in Proposition \ref{pr:chi:symmetry} that $\chi(f)$ is symmetric,
 and we know from \cite[Proposition 1(2)]{Lechner03} that $\phi(f)$ is symmetric.
 Therefore, the sum $\fct(f) = \phi(f) + \chi(f)$ is symmetric as well.
 
 (\ref{pr:fct:kg}) This follows from the facts that $((\Box + m^2)f)^+ = 0$, as in \cite[Proposition 1(3)]{Lechner03}
 and that $f$ appears only through $f^+$ in the definition of $\chi(f)$.
 
 (\ref{pr:fct:covariance}) We showed the covariance of $\chi(f)$ in Proposition \ref{pr:chi:covariance} and
 the covariance of $\phi(f)$ is shown in \cite[Proposition 1(4)]{Lechner03}, therefore
 the sum $\chi(f) = \phi(f) + \chi(f)$ is covariant, too. 
\end{proof}
Of course, similar properties hold for $\fct'(g)$.

Compared with \cite[Proposition 1]{Lechner03}, our fields are of a subtler nature.
The operator $\fct(f)$ does not preserve its domain (see Appendix \ref{domain}),
neither is essentially self-adjoint.

\subsubsection*{Weak commutativity}
Here we prove our main result.
As the domains of $\fct(f)$ and $\fct'(g)$ are subtle, we cannot always form the products
$\fct(f)\fct'(g)$ and $\fct'(g)\fct(f)$, and the commutator $[\fct(f),\fct'(g)]$.
Instead, we consider the weak form of commutation.
Consider the follogin linear space of vectors:
\[
\left\{
\begin{array}{c|l}
 & \textstyle{\prod_j S\left(\theta-\theta_j + \frac{\pi i}3\right)\Psi_n(\theta,\theta_1,\cdots,\theta_{n-1})} \mbox{ and}\\
 \Psi \in \dom(\fct(f))\cap\dom(\fct'(g)) &\textstyle{
 \prod_j S\left(\theta-\theta_j + \frac{2\pi i}3\right)\Psi_n(\theta,\theta_1,\cdots,\theta_{n-1})} \mbox{ have }L^2(\RR^{n-1})\mbox{-} \\
 &\textstyle{ \mbox{valued bounded analytic continuations in }\theta \mbox{ to } \theta\pm \epsilon i,}\\
 &\textstyle{ \mbox{for some } \epsilon > 0}
\end{array}
\right\}
\]
The intersection $\dom(\fct(f))\cap\dom(\fct'(g))$ has an interesting property.
The $n$-particle component $\Psi_n$ of a vector $\Psi$ has an analytic continuation in $\theta_1$ in the
negative imaginary part, but by $S$-symmetry, it can be analytically continued in any variable
(with possible poles at the poles of $S$). At the same time, it has analytic continuation in all the
variables in the positive imaginary part, as it is in $\dom(\fct'(g))$.
Therefore, $\Psi_n$ is not only the $L^2$ boundary value
of an analytic function, but it has a continuous value at $\theta_j\in \RR$.

Now, by \ref{fermionic}, $S(0) = -1$ and $\Psi_n$ is $S$-symmetric, therefore $\Psi_n(\theta_1,,\cdots, \theta_n)$
has a zero at $\theta_j -\theta_k =0$.
In the expression below,
\[
\prod_j S\left(\theta-\theta_j + \frac{\pi i}3\right)\Psi_n(\theta,\theta_1,\cdots,\theta_{n-1}),
\]
the $S$-factors have simple poles at $\theta - \theta_j = 0$, and they are cancelled
by the zeros of $\Psi_n(\theta,\theta_1,\cdots,\theta_n)$. Therefore, it still has pointwise meaning.
We conjecture that the additional condition that it has an $L^2$-bounded analytic continuation
in $\theta$ should be automatic, and the linear space above should be simply the intersection
$\dom(\chi(f))\cap\dom(\chi'(g))$.
Yet, as the domains are a subtle question, we content ourselves in the present work by
showing the weak commutativity on this space.
One can easily check that it is dense in $\H$
by considering the vectors constructed in Proposition \ref{pr:chi:symmetry}.

For vectors $\Phi,\Psi$ in the above space,
we show weak commutativity, i.e., $\<\fct(f)\Phi, \fct'(g)\Psi\> = \<\fct'(g)\Phi, \fct(f)\Psi\>$.
 From weak commutativity and some estimate of operators, strong commutativity can follow (e.g.\! \cite{DF77}).
Unfortunately, we have not yet been able to establish such estimates and in this work we restrict ourselves
to these weak relations. See also Appendix \ref{domain}.

\begin{theorem}\label{theo:commutator}
 Let $f$ and $g$ be real test functions supported in $W_\L$ and $W_\R$, respectively.
 Then, for each $\Phi, \Psi$ in the linear space above,
 it holds that
 \[
 \<\fct(f)\Phi, \fct'(g)\Psi\> = \<\fct'(g)\Phi, \fct(f)\Psi\>.
 \]
\end{theorem}
\begin{proof}
 Recall that the vectors $\Phi, \Psi$ have finitely many non-zero components and our operators acts
 componentwise:
 \begin{align*}
  \fct(f) &= \phi(f) + \chi(f) = z^\dagger(f^+) + \chi(f) + z(J_1f^-), \\
  \fct'(g) &= \phi'(g) + \chi'(g) = z'^\dagger(g^+) + \chi'(g) + z'(J_1g^-).
 \end{align*}
 In this proof, it should be kept in mind that $\Phi$ and $\Psi$ are already $S$-symmetric.
 
 Let us compute the commutator $[ \fct(f), \fct'(g)]$, which by the above expands into several terms that we will compute individually.
 \begin{flushleft} 
 {\it The commutator $[\chi(f), z'(J_1g^-)]$}
 \end{flushleft}
 We can actually compute this commutator in the operator (not the weak) form,
 as we do not encounter the problem of domains. Recall that, if $g$ is real,
 $J_1g^- = g^+$.
 
 Let us look at the expression we derived early in Section \ref{chi},
 in which $\chi_n(f)$ is written as a sum of $n$ terms.
 Its first $n-1$ terms act in the first $n-1$ variables in the same way as
 $n-1$ terms in $\chi_{n-1}$ do. On the other hand, $z'(g^+)$ acts only
 on the $n$-th variable $\theta_n$.
 Therefore, the only remaining term in this commutator is the one
 which comes from the $n$-th term in the expression, which is
  \begin{align*}
  & ([\chi(f),z'(J_1g^-)]\Psi_n)(\theta_1,\cdots,\theta_{n-1}) \\
  &= i\eta\sqrt n\int d\theta\, \overline{g^+(\theta)}\prod_{1\le j \le n-1} S\left(\theta-\theta_j + \frac{\pi i}3\right)
 f^+\left(\theta + \frac{\pi i}3\right) \Psi_n\left(\theta_1,\cdots,\theta_{n-1},\theta - \frac{\pi i}3\right).
 \end{align*}
 For the later convenience, we further rewrite this using the $S$-symmetry of $\Psi_n$:
 \begin{align*}
  & ([\chi(f),z'(J_1g^-)]\Psi_n)(\theta_1,\cdots,\theta_{n-1}) \\
  = i\eta\sqrt n \int & d\theta\, \overline{g^+(\theta)}\prod_{1\le j \le n-1} S\left(\theta - \theta_j + \frac{\pi i}3\right)
 f^+\left(\theta + \frac{\pi i}3\right) \\
 & \;\;\;\times S\left(\theta - \theta_j - \frac{\pi i}3\right)\Psi_n\left(\theta - \frac{\pi i}3,\theta_1,\cdots,\theta_{n-1}\right) \\
  = i\eta\sqrt n\int & d\theta\, \overline{g^+(\theta)}\prod_{1\le j \le n-1} S(\theta - \theta_j) 
 f^+\left(\theta + \frac{\pi i}3\right) \Psi_n\left(\theta - \frac{\pi i}3, \theta_1,\cdots,\theta_{n-1}\right) \\
  = i\eta\sqrt n\int & d\theta\, \overline{g^+\left(\theta - \frac{\pi i}3\right)}\prod_{1\le j \le n-1} S\left(\theta - \theta_j + \frac{\pi i}3\right) 
 f^+\left(\theta + \frac{2\pi i}3\right) \Psi_n\left(\theta, \theta_1,\cdots,\theta_{n-1}\right) \\
  =i\eta\sqrt n\int & d\theta\, g^+\left(\theta - \frac{2\pi i}3\right)\prod_{1\le j \le n-1} S\left(\theta - \theta_j + \frac{\pi i}3\right) 
 f^+\left(\theta + \frac{2\pi i}3\right) \Psi_n\left(\theta, \theta_1,\cdots,\theta_{n-1}\right),
 \end{align*}
 where in the second equality we used the bootstrap equation, in the third equality the assumed property of $\Psi$ explained before this Theorem \ref{theo:commutator}
 and Lemma \ref{lm:l2shift}, and the last equality follows because $g$ is real.

 \begin{flushleft} 
 {\it The commutator $[z(J_1f^-), \chi'(g)]$}
 \end{flushleft}
 This can be computed in a similar way as before. By paying attention that $z$ reduces the number of variables and
 shifts the indices of the remaining ones, the result is:
 \begin{align*}
  & [(z(J_1f^-), \chi'(g)]\Psi_n)(\theta_1,\cdots,\theta_{n-1}) \\
  =\; -i\eta\sqrt n\int &d\theta\, \overline{f^+(\theta)}\prod_{1\le j \le n-1} S\left(\theta_j-\theta + \frac{\pi i}3\right)
 g^+\left(\theta - \frac{\pi i}3\right) \Psi_n\left(\theta + \frac{\pi i}3,\theta_1,\cdots,\theta_{n-1}\right) \\
  =\; -i\eta\sqrt n\int &d\theta\, \overline{f^+\left(\theta + \frac{\pi i}3\right)}\prod_{1\le j \le n-1} S\left(\theta_j-\theta + \frac{2\pi i}3\right) 
 g^+\left(\theta - \frac{2\pi i}3\right) \Psi_n\left(\theta,\theta_1,\cdots,\theta_{n-1}\right) \\
  =\; -i\eta\sqrt n\int &d\theta\, f^+\left(\theta + \frac{2\pi i}3\right)\prod_{1\le j \le n-1} S\left(\theta-\theta_j + \frac{\pi i}3\right)
 g^+\left(\theta - \frac{2\pi i}3\right) \Psi_n\left(\theta,\theta_1,\cdots,\theta_{n-1}\right),
 \end{align*}
 where we used the assumed domain property of $\Psi$, Lemma \ref{lm:l2shift},
 the crossing symmetry of $S$, and that $f$ is real.
 This coincides with the result of the commutator $[\chi(f), z'(J_1g^-)]$ up to a sign,
 therefore they cancel each other.
 
 \begin{flushleft} 
 {\it The commutators $[z^\dagger(f^+), \chi'(g)]$ and $[\chi(f), z'^\dagger(g^+)]$}
 \end{flushleft}
 One can show that these two commutators cancel each other by taking adjoints and repeat the
 computation as in the commutators before.
 More precisely, before we have computed weak commutators such as $\<K\Psi, Z\Phi\> - \<Z^*\Psi, K\Phi\>$,
 where  $K$ corresponds to either $\chi(f)$ or $\chi'(g)$ and $Z$ is the annihilation operators
 $z(J_1f^-)$ or $z'(J_1g^-)$.
 Here we must compute numbers such as $\<K\Psi, Z^*\Phi\> - \<Z\Psi, K\Phi\>$,
 which is complex conjugate to the previous one.

 \begin{flushleft} 
 {\it The commutator $[\phi(f), \phi'(g)]$}
 \end{flushleft}
 This part has been essentially done in \cite[P.13]{Lechner03}. The difference is that
 $S$ has two poles in the physical strip, therefore we obtain residues when we shift the integration contour.
 This commutator preserves the particle number, therefore, it suffices to compute its action
 at fixed $n$.
 Let us take $\Psi, \Phi$ with only $n$ particle components.
 By considering that the poles of $S$ are simple, the result is (note that we switched $f$ and $g$ from \cite{Lechner03})
 \begin{align*}
 & ([\phi(f), \phi'(g)]\Psi_n)(\theta_1,\cdots,\theta_n)\\
 =& -\int d\theta\, \left(g^-(\theta)f^+(\theta)\prod_{j=1}^n S(\theta-\theta_j) - g^-(\theta + \pi i)f^+(\theta + \pi i)\prod_{j=1}^n S(\theta-\theta_j + \pi i) \right)\\
 & \;\;\;\times \Psi_n(\theta_1,\cdots,\theta_n) \\
 =-2\pi i& \Bigg(\sum_{k = 1}^n Rg^-\left(\theta_k + \frac{2\pi i}3\right)f^+\left(\theta_k + \frac{2\pi i}3\right)
  \prod_{j\neq k} S\left(\theta_k-\theta_j + \frac{2\pi i}3\right) \\
 & - \sum_{k = 1}^n Rg^-\left(\theta_k + \frac{\pi i}3\right)f^+\left(\theta_k + \frac{\pi i}3\right)
  \prod_{j\neq k} S\left(\theta_k-\theta_j + \frac{\pi i}3\right)\Bigg) \times \Psi_n(\theta_1,\cdots,\theta_n) \\
 =-2\pi i& \Bigg(\sum_{k = 1}^n Rg^+\left(\theta_k - \frac{\pi i}3\right)f^+\left(\theta_k + \frac{2\pi i}3\right)
  \prod_{j\neq k} S\left(\theta_k-\theta_j + \frac{2\pi i}3\right) \\
 & - \sum_{k = 1}^n Rg^+\left(\theta_k - \frac{2\pi i}3\right)f^+\left(\theta_k + \frac{\pi i}3\right)
  \prod_{j\neq k} S\left(\theta_k-\theta_j + \frac{\pi i}3\right)\Bigg) \times \Psi_n(\theta_1,\cdots,\theta_n),
 \end{align*}
 where we used that the residue of $S$ at $\frac{\pi i}3$ is $-R$, which follows from the crossing symmetry.
 In order to justify the second equality, we should note that if any pair of $\theta_k$'s does not coincide,
 then the integrand has a simple pole and we have the equality. The complement of such $\pmb{\theta}$'s has Lebesgue measure zero, therefore we obtain the equality. Actually, although
 this expression looks unbounded as a function of $\theta_k$ because of the poles of $S$, it is actually bounded since these poles at $\theta_k - \theta_j = 0$
cancel each other. (This is not surprising, since
 $\phi(f)$ and $\phi'(g)$ are bounded on $\H_n$.)

 \begin{flushleft} 
 {\it The commutator $[\chi(f), \chi'(g)]$}
 \end{flushleft}
 Now we come to the most important part of the computations. As both $\chi(f)$ and $\chi'(g)$
 preserves each $n$-particle space, we assume again that $\Psi, \Phi$ have only $n$-particle components.
 
 Let us recall the expressions of $\chi(f), \chi'(g)$ we derived in Section \ref{chi}:
 \begin{equation*}
 \begin{aligned}
  \chi_n(f) &= \sum_{1\le k \le n} D_n(\rho_k)(\chi_1(f)\otimes\1\otimes\cdots\otimes\1)P_n, \\
  \chi'_n(f) &= \sum_{1\le j \le n} D_n(\rho'_j)(\1\otimes\cdots\otimes\1\otimes\chi'_1(g))P_n.
 \end{aligned}
 \end{equation*}
 Therefore, the term $\<\chi'(g)\Phi, \chi(f)\Psi\>$ of the commutator is
 \[
   \sum_{j,k}\<D_n(\rho'_j)(\1\otimes\cdots\otimes\1\otimes\chi'_1(g))\Phi, D_n(\rho_k)(\chi_1(f)\otimes\1\otimes\cdots\otimes\1)\Psi\>.
 \]
 Actually, $\rho_k$ here can be replaced by any permutation $\sigma$ such that
 $\sigma(1) = k$. Let us choose the transposition $\tau_{1,k}: (1,k) \mapsto (k,1)$.
 $\tau_{1,1}$ coincides with the unit element of $\mathfrak{S}_n$.
 Similarly, we can use the transposition $\tau_{n-j+1,n}$ instead of $\rho'_j$. Then it is clear that $\tau_{1,k}$
 and $\tau_{n-j+1,n}$ commute unless $k=n$ or $j=n$ or $k=n-j+1$.
 For pairs of $k$ and $j$ such that the two transpositions commute, the scalar product reduces to the following:
 \begin{equation}
 \begin{aligned}\label{tautwo}
  &\<D_n(\tau_{n-j+1,n})(\1\otimes\cdots\otimes\1\otimes\chi'_1(g))\Phi, D_n(\tau_{1,k})(\chi_1(f)\otimes\1\otimes\cdots\otimes\1)\Psi\> \\
  & = \<(\1\otimes\cdots\otimes\1\otimes\chi'_1(g))\Phi, (\chi_1(f)\otimes\1\otimes\cdots\otimes\1)\Psi\> \\
  & = \<(\chi_1(f)\otimes\1\otimes\cdots\otimes\1)\Phi, (\1\otimes\cdots\otimes\1\otimes\chi'_1(g))\Psi\>,
 \end{aligned}
 \end{equation}
 where the last equality follows by Lemma \ref{lm:abstractshift} (by writing $\chi_1(f) = \sqrt{2\pi|R|}x_f\Delta_1^\frac16$, etc.).
 If $k=n, j\neq 1,n$ (respectively $k\neq 1,n$ and $j=n$), one has $\tau_{n-j+1,n}\tau_{1,n} = \tau_{1,n-j+1}\tau_{n-j+1,n}$
 (resp. $\tau_{1,n}\tau_{1,k}= \tau_{1,k}\tau_{n,k}$), and we can reduce these contributions to the same value \eqref{tautwo}. 
 The case $k=j=n$ is also easy and gives the same contribution. We will see that they get cancelled by the term $\<\chi(f) \Phi, \chi'(g)\Psi\>$ of the commutator $[\chi(f),\chi'(g)]$.

 The remaining terms are those with $k=n-j+1$. Here, $\chi_1(f)$ and $\chi'_1(g)$ act on the same variable:
 \begin{align*}
  &\<D_n(\rho'_{n-k+1})(\1\otimes\cdots\otimes\1\otimes\chi'_1(g))\Phi, D_n(\rho_k)(\chi_1(f)\otimes\1\otimes\cdots\otimes\1)\Psi\> \\
  = 2\pi|R|\int d\theta_1\cdots d\theta_n\, &\prod_{l=1}^{k-1} S\left(\theta_k - \theta_l + \frac{\pi i}3\right) f^+\left(\theta_k + \frac{\pi i}3\right)\Psi_n\left(\theta_1,\cdots,\theta_k-\frac{\pi i}3,\cdots,\theta_n\right) \\
 \times &\overline{\prod_{l=k+1}^{n} S\left(\theta_l - \theta_k + \frac{\pi i}3\right) g^+\left(\theta_k - \frac{\pi i}3\right)\Phi_n\left(\theta_1,\cdots,\theta_k+\frac{\pi i}3,\cdots,\theta_n\right)} \\
  = -2\pi iR\int d\theta_1\cdots d\theta_n\, &\prod_{l=1}^{k-1} S\left(\theta_k - \theta_l + \frac{2\pi i}3\right) f^+\left(\theta_k + \frac{2\pi i}3\right)\Psi_n\left(\theta_1,\cdots,\theta_n\right) \\
  &\times \overline{\prod_{l=k+1}^{n} S\left(\theta_l - \theta_k + \frac{2\pi i}3\right) g^+\left(\theta_k - \frac{2\pi i}3\right)\Phi_n\left(\theta_1,\cdots,\theta_n\right)} \\
  = -2\pi iR\int d\theta_1\cdots d\theta_n\, &\prod_{l=1}^{k-1} S\left(\theta_k - \theta_l + \frac{2\pi i}3\right) f^+\left(\theta_k + \frac{2\pi i}3\right)\Psi_n\left(\theta_1,\cdots,\theta_n\right) \\
 & \times \prod_{l=k+1}^{n} S\left(\theta_k - \theta_l + \frac{2\pi i}3\right) g^+\left(\theta_k - \frac{\pi i}3\right)\overline{\Phi_n\left(\theta_1,\cdots,\theta_n\right)} \\
  = -2\pi iR\int d\theta_1\cdots d\theta_n\, &\prod_{l\neq k} S\left(\theta_k - \theta_l + \frac{2\pi i}3\right) f^+\left(\theta_k + \frac{2\pi i}3\right)\Psi_n\left(\theta_1,\cdots,\theta_n\right)
    g^+\left(\theta_k - \frac{\pi i}3\right) \\
 &\overline{\Phi_n\left(\theta_1,\cdots,\theta_n\right)},
 \end{align*}
 where we used the assumption that the product of $\Psi, \Phi$ and the $S$ factor remains $L^2$, Lemma \ref{lm:l2shift}
 and the condition \ref{poles-boundedness} which implies that $R = i|R|$. 
 
 Note that this is the second term in the commutator $[\chi(f), \chi'(g)]$, therefore it gets another minus sign.
 The expression so obtained is equal to the first contribution from the commutator $[\phi(f),\phi'(g)]$ with reversed sign and
 coupling with $\Phi$, therefore, they cancel each other.

 Let us examine the term $\<\chi(f) \Phi, \chi'(g)\Psi\>$ in the commutator $[\chi(f),\chi'(g)]$.
 The computation is similar to the previous case and we obtain the following.
 As before, for pairs of $k$ and $j$ such that $k \neq n-j+1$, one has the contributions:
 \begin{align*}
  &\<D_n(\tau_{n-j+1,n})(\1\otimes\cdots\otimes\1\otimes\chi'_1(g))\Phi, D_n(\tau_{1,k})(\chi_1(f)\otimes\1\otimes\cdots\otimes\1)\Psi\> \\
  & = \<(\1\otimes\cdots\otimes\1\otimes\chi'_1(g))\Phi, (\chi_1(f)\otimes\1\otimes\cdots\otimes\1)\Psi\> \\
  & = \<(\chi_1(f)\otimes\1\otimes\cdots\otimes\1)\Phi, (\1\otimes\cdots\otimes\1\otimes\chi'_1(g))\Psi\>,
 \end{align*}
 and this cancels the contribution in Equation \eqref{tautwo}.
 
 The remaining terms in $\<\chi(f) \Phi, \chi'(g)\Psi\>$ are:
 \begin{align*}
  &\<D_n(\rho_k)(\chi_1(f)\otimes\1\otimes\cdots\otimes\1)\Phi, D_n(\rho'_{n-k+1})(\1\otimes\cdots\otimes\1\otimes\chi'_1(g))\Psi\> \\
  = -2\pi iR\int d\theta_1\cdots d\theta_n\, &\prod_{l=k+1}^{n} S\left(\theta_l - \theta_k + \frac{\pi i}3\right) g^+\left(\theta_k - \frac{\pi i}3\right)\Psi_n\left(\theta_1,\cdots,\theta_k+\frac{\pi i}3,\cdots,\theta_n\right) \\
 & \times \prod_{l=1}^{k-1} \overline{S\left(\theta_k - \theta_l + \frac{\pi i}3\right) f^+\left(\theta_k + \frac{\pi i}3\right)\Phi_n\left(\theta_1,\cdots,\theta_k-\frac{\pi i}3,\cdots,\theta_n\right)} \\
  = -2\pi iR\int d\theta_1\cdots d\theta_n\, &\prod_{l=k+1}^{n} S\left(\theta_l - \theta_k + \frac{2\pi i}3\right) g^+\left(\theta_k - \frac{2\pi i}3\right)\Psi_n\left(\theta_1,\cdots,\theta_k,\cdots,\theta_n\right) \\
 & \times \prod_{l=1}^{k-1} \overline{S\left(\theta_k - \theta_l + \frac{2\pi i}3\right) f^+\left(\theta_k + \frac{2\pi i}3\right)\Phi_n\left(\theta_1,\cdots,\theta_k,\cdots,\theta_n\right)} \\
  = -2\pi iR\int d\theta_1\cdots d\theta_n\, &\prod_{l=k+1}^{n} S\left(\theta_l - \theta_k + \frac{2\pi i}3\right) g^+\left(\theta_k - \frac{2\pi i}3\right)\Psi_n\left(\theta_1,\cdots,\theta_n\right) \\
 & \times \prod_{l=1}^{k-1} S\left(\theta_l - \theta_k + \frac{2\pi i}3\right) f^+\left(\theta_k + \frac{\pi i}3\right)\overline{\Phi_n\left(\theta_1,\cdots,\theta_n\right)} \\
  = -2\pi iR\int d\theta_1\cdots d\theta_n\, &\prod_{l\neq k}^{n} S\left(\theta_k - \theta_l + \frac{\pi i}3\right) g^+\left(\theta_k - \frac{2\pi i}3\right)\Psi_n\left(\theta_1,\cdots,\theta_n\right)
  f^+\left(\theta_k + \frac{\pi i}3\right) \\
 &\times \overline{\Phi_n\left(\theta_1,\cdots,\theta_n\right)}.
 \end{align*}
The last expression is equal to the second contribution from the weak form of the
commutator $\<\Phi, [\phi(f),\phi'(g)]\Psi\>$ up to the sign, therefore, they cancel each other. 
 
Altogether, we have seen that all the terms in $[\fct(f),\fct'(g)]$ cancel.
\end{proof}

\subsubsection*{Reeh-Schlieder property}
As the fields $\fct(f)$ and $\fct'(g)$ do not preserve their domains, especially one cannot iterate them on the vacuum
more than once (see Appendix \ref{domain}), the Reeh-Schlieder property does not hold for polynomials of these fields.
Instead, if we assume the existence of nice self-adjoint extensions, 
we can argue as below. But as we do not have such extensions, we refrain from stating it as a theorem.

Let us suppose that for each $g$ there is a self-adjoint extension of $\fct'(g)$, which we denote by the same symbol,
such that $\fct'$ is covariant with respect to $U$.
Suppose also that, for each $f$, $\fct(f)$ has a nice self-adjoint extension, such that
$\fct'(g)$ and $\fct(f)$ strongly commute.
We consider the von Neumann algebra
\[
 \M = \{e^{i\fct'(g)}: \supp g \subset W_\R\}'',
\]
and we have to show that $\overline{\M\Omega} = \H$.
Actually, as $\M$ is an algebra of bounded operators containing the identity operator $\1$,
we can freely use the fact that $\overline{\M\Omega} = \M\overline{\M\Omega}$.

Take first the one-particle space. We have $\fct'(g)\Omega = g^+ \in \H_1$ and this is
in the above closure because $\fct'(g)\Omega = \frac{d}{dt}e^{it\fct'(g)}\Omega$ and
$\Omega$ is in the domain of $\fct(f)$ \cite[Theorem VIII.7]{RSII}.
By the one-particle Reeh-Schlieder property (e.g. \cite[Theorem 3.2.1]{Longo08}),
it follows that $(\CC\Omega \oplus \H_1) \subset \overline{\M\Omega}$.

Note that $\H_1 \cap \dom(\fct'(g))$, which includes $\H_1 \cap \dom(\chi'(g))$,
has a dense subspace. For any such vector $\xi \in \H_1 \cap \dom(\chi'(g))$, it holds
$\fct'(g)\xi = \frac{d}{dt}e^{it\fct'(g)}\xi \in (\CC\Omega \oplus \H_1 \oplus \H_2)$
and their projection to $\H_2$ is dense in $\H_2$ (as it comes from the action of $z^\dagger$).
We can subtract the $(\CC\Omega \oplus \H_1)$-component since $\M$ is an algebra,
and obtain that $\H_2$-component belongs to $\M\overline{\M\Omega} = \overline{\M\Omega}$.
Therefore, it follows that $(\CC\Omega \oplus \H_1 \oplus \H_2) \subset \overline{\M\Omega}$.

The rest is shown by induction: Assume that $(\CC\Omega \oplus \cdots \oplus \H_n) \subset \overline{\M\Omega}$.
By differentiation, we obtain $\fct'(g)\Psi = \frac{d}{dt}e^{it\fct'(g)}\Psi \in (\CC\Omega \oplus \cdots \oplus \H_{n+1})$
and we can extract the $\H_{n+1}$-component. Such $\H_{n+1}$-components are obtained by $z^\dagger$, thus
form a dense subspace in $\H_{n+1}$. Namely, we showed that
$(\CC\Omega \oplus \cdots \oplus \H_{n+1}) \subset \overline{\M\Omega}$,
which completes the induction.

\subsubsection*{(Non-)temperateness of the fields}
Let us assume that we have a Haag-Kastler net.
Borchers, Buchholz and Schroer \cite{BBS01} called an operator $G$ which is affiliated to the wedge-algebra and
generates a one-particle state from the vacuum a {\bf polarization-free generator}.
A polarization-free generator is said to be {\bf temperate}, if there is a translation-invariant
dense domain of $G$ such that, for any vector $\Psi$ in that domain,
$a\mapsto GU(a,0)\Psi$ is strongly continuous and polynomially bounded.
The existence of temperate polarization-free generators restrict drastically the possibility of interaction.

We argue that the closure of our field $\fct(f)$ (which we denote by the same symbol)
has no such polynomially growing vector.
Indeed, let $\Psi$ be a vector in the domain of the closure of $\fct(f)$.
Even if $\Psi$ is only in the domain of the closure of $\fct(f) = \phi(f) + \chi(f)$, the operator $\phi(f)$
is continuous on each space of fixed particle number $\H_n$, thus each $n$-particle component $\Psi_n$ of
$\Psi$ is in the domain of $\chi_n(f)$,
and $\fct(f)$ can be computed componentwise.
If we look at the one-particle component of $\chi(f)U(a,0)\Psi$, we get
\[
 (\chi_1(f)U_1(a,0)\Psi_1)(\theta) = \sqrt{2\pi |R|}f^+\left(\theta + \frac{\pi i}3\right)e^{ia\cdot p\left(\theta - \frac{\pi i}3\right)}
 \Psi_1\left(\theta - \frac{\pi i}3\right).
\]
If $a$ tends to the negative spacelike direction, this grows exponentially.
As we remarked, the contribution from $\phi(f)$ is bounded, hence this shows that $\fct(f)U(a,0)\Psi$ is
not polynomially bounded.

\section{The form factor program and polarization-free generators}\label{sec:formfactor}
In this section, we discuss the connection of our approach to the form factor program \cite{Smirnov92, BK04}.
We will also present how to derive our new term $\chi(f)$ under certain assumptions.

Both in the form factor program and in the operator-algebraic approach,
we have to construct local observables. In the form factor program,
the matrix elements of local operators are considered and they are subject to
several conditions, the form factor axioms, and one has to find solutions (form factors)
of these conditions. On the other hand, in the operator-algebraic approach,
it has been proved \cite{BBS01} that, in any model, there are polarization-free generators, namely
operators which are localized in a wedge-region and generate one-particle states from the
vacuum. Therefore, if a Haag-Kastler net exists for an integrable QFT, there must be polarization-free generators.
As they are localized in a wedge, they must commute with operators localized in the
causal complement of the wedge, especially with local operators which come from form factors.
Here we perform such formal computations.

\subsection{Compatibility with the form factor program}\label{compatibility}
\subsubsection*{Operator expansion for analytic S-matrix}
In \cite{BC14}, we investigated the structure of strictly local observables in integrable quantum field theories in $1+1$ dimensions in the case where the theory has only one species of massive scalar particle and the two-particle scattering function is analytic in the physical strip.

Specifically, we extracted more information on the properties of these local observables $A$ by expanding them into a series of normal-ordered strings of $z^\dagger,z$,
\begin{equation}\label{exp}
A= \sum_{m,n =0}^\infty \int \frac{d\pmb{\theta}d\pmb{\eta}}{m!n!} f_{m,n}^{[A]}(\pmb{\theta},\pmb{\eta})z^\dagger(\theta_1)\cdots  z^\dagger(\theta_m)z(\eta_1)\cdots z(\eta_n).
\end{equation}
It was shown in \cite{BC13} that this expansion holds for every operator or quadratic form $A$ in a certain regularity class and it is independent of the localization region of $A$. It is similar to the well-known form factor expansion, although it is not identical to it; for a pointlike field $A$, the definition of
the coefficients $f_{m,n}^{[A]}(\pmb{\theta},\pmb{\eta})$ formally agree with the form factors for certain regions of the arguments. More specifically,
following \cite[Section 3]{BK04}\cite[p.11]{Quella99} we have that
\begin{equation*}
f_{0,n}^{[A]}(\pmb{\theta}) = \langle \Omega, A z^\dagger(\theta_n) \ldots z^\dagger(\theta_1) \Omega \rangle = F^A(\theta_1, \ldots, \theta_n)
\end{equation*}
for $\theta_n > \ldots > \theta_1$ and where $F^A$ denotes the form factor of $A$. For the other coefficients $f_{m,n}^{[A]}(\pmb{\theta},\pmb{\eta})$, there is an explicit expression in term of vacuum expectation values of $A$ \cite{BC13}, and it was shown in \cite{BC14} that if $A$ is a general observable localized in a bounded region, then the $f_{m,n}^{[A]}$ are boundary values of a common meromorphic function, i.e., one has
\begin{equation*}
f_{m,n}^{[A]}(\pmb{\theta},\pmb{\eta}) = F_{m+n}(\pmb{\theta}+i\pmb{0}, \pmb{\eta} +i\pmb{\pi}-i\pmb{0}).
\end{equation*}
Choosing the bounded region to be a standard double cone of radius $r>0$ around the origin, these $F_k(\pmb{\zeta})$ ($\pmb{\zeta} \in \mathbb{C}^k$) fulfill the following properties, which we write informally here (the exact statement of these properties can be found in \cite{BC14}.)
\begin{enumerate}
\renewcommand{\theenumi}{(F\arabic{enumi})}
\renewcommand{\labelenumi}{\theenumi}
\item \label{mero} They are meromorphic on $\mathbb{C}^k$. In particular, except for possible first-order poles at $\zeta_n -\zeta_m =i\pi$ (``kinematic poles''), they are analytic on the tube region $\operatorname{Im} \zeta_1 < \ldots < \operatorname{Im} \zeta_k < \operatorname{Im} \zeta_1 +2\pi$.

\item \label{sym} They are $S$-symmetric, that is, for all complex arguments $\zeta_1, \ldots, \zeta_k$, the following relation holds:
\begin{equation*}
F_k(\zeta_1, \ldots, \zeta_{j+1},\zeta_j, \ldots, \zeta_k) = S(\zeta_j -\zeta_{j+1}) F_k(\zeta_1, \ldots, \zeta_j, \zeta_{j+1}, \ldots, \zeta_k).
\end{equation*}

\item \label{period} They are $S$-periodic, i.e.,
\begin{equation*}
F_k(\zeta_1, \ldots, \zeta_{k-1}, \zeta_k +2i\pi) = F_k(\zeta_k, \zeta_1, \ldots, \zeta_{k-1})
\end{equation*}
(a similar property holds also in the other variables.)

\item \label{res} The value of their residue at $\zeta_k -\zeta_1= i\pi$ is given by
\begin{equation*}
\operatorname*{res}_{\zeta_k -\zeta_1 = i\pi}F_k(\pmb{\zeta}) = \frac{1}{2\pi i} \Big( 1- \prod_{p=1}^k S(\zeta_1 -\zeta_p)\Big) F_{k-2}(\zeta_2, \ldots, \zeta_{k-1}) 
\end{equation*}
(and a similar formula holds for the residues at the other kinematic poles.)

The remaining properties concern the bounds that these functions fulfill, namely
\item \label{bouedge}  $\pmb{\theta} \rightarrow F_k(\pmb{\theta}+i\pmb{0})$ is square integrable.

\item \label{bouinter}At real infinity the growth behaviour of $F_k$ is essentially given by
\begin{equation*}
\lvert F(\pmb{\theta} +i\pmb{\lambda})\rvert \sim \prod_{j=1}^k e^{m r \cosh \theta_j |\sin \lambda_j|}.
\end{equation*}
\end{enumerate}

These relations between form factors and local observables can be made precise and suitable variations of \ref{mero}--\ref{bouinter} holds for  $F_k$ \emph{if and only if} $A$ is localized in a double cone (see \cite[Theorem  5.4]{BC14}).

The functions $F_k$ have a rich pole structure: Apart from the so-called kinematic poles, they have further singularities on hyperplanes in $\mathbb{C}^k$, which are due to the poles of the scattering function ``outside the physical strip''.

\subsubsection*{Cases with poles in the physical strip}
In the present paper, the scattering function has additional poles in the physical strip. We will therefore expect that properties \ref{mero}--\ref{bouinter} need to be modified to take into account these extra poles of the S-matrix. We call these new properties (P1)--(P6). 
In particular, for (P1) we demand that the functions $F_k$ are meromorphic with the pole structure as in \ref{mero} and additional poles at $\zeta_n -\zeta_m = \frac{2\pi i}{3}$ with $n > m$. The presence of these additional poles, which are a consequence of the poles of the S-matrix inside the physical strip, is a well-known feature in the form factor program \cite{BabujianFoersterKarowski:2006}.

Further, we expect that properties (P2) and (P3) can stand unmodified from \ref{sym} and \ref{period}, while (P5), which is required for the well-definedness of the expansion \eqref{exp}, will possibly need small modifications from \ref{bouedge}; these depend on details of the domain of definition of \eqref{exp} in a theory with bound states, which we do not enter here.

We will also need to properly adapt \ref{bouinter} to account for the extra poles of $F_k(\pmb{\zeta})$ at $\zeta_n -\zeta_m = \frac{2\pi i}{3}$ with $n>m$, yielding a new condition (P6). 

However the essential new ingredient to be added to \cite[Definition 5.3]{BC14} is described as follows:
\begin{enumerate}
\renewcommand{\theenumi}{(P\arabic{enumi})}
\renewcommand{\labelenumi}{\theenumi}
\setcounter{enumi}{6}
\item \label{resbou} The $F_k$ have first order poles at $\zeta_n -\zeta_m = \frac{2\pi i}{3}$, where $1\leq m < n \leq k$ and one has
\begin{equation*}
\operatorname*{res}_{\zeta_2 -\zeta_1 = \frac{2\pi i}{3}} F_k (\pmb{\zeta}) = \frac{\eta}{2\pi }F_{k-1}\Big(\zeta_1 +\frac{\pi i}{3}, \zeta_3, \ldots, \zeta_k \Big);
\end{equation*}
the residues at the other poles can again be inferred from \ref{sym}.
\end{enumerate}
A similar formula can be found in the form factor program \cite{BabujianFoersterKarowski:2006} up to a constant factor, where it characterizes the form factors of point-like localized fields in models with bound states.

\subsubsection*{Sample computations}
We would like to show that if the form factors $F_k$ of an observable $A$ fulfill (P1)--\ref{resbou} then $A$ is local with respect to the wedge local field $\fct$, namely $[A,\fct(f)]=0$ and $[A,\fct'(g)]=0$ where $f$ and $g$ are suitably regular functions (c.f.\! \cite[P.13]{Lechner03}, Proposition \ref{pr:chi:symmetry} and Theorem \ref{theo:commutator}) with supports in $W_\L -(0,r)$ and $W_\R +(0,r)$, respectively.

Here we will only show that $[A, \fct'(g)]= [A, \chi'(g)] +  [A, \phi'(g)]=0$ in matrix elements between one-particle states $\Psi,\Phi$ on the level of formal computation. 

We compute the commutator $ \langle \Phi, [A, \chi'(g)] \Psi \rangle =  \langle \Phi, A \chi'(g) \Psi \rangle - \langle \Phi , \chi'(g) A \Psi \rangle$ using the actions of  \eqref{exp} and of $\chi'(g)$ derived in Section \ref{chi} on the vectors $\Psi, \Phi$. We obtain
\begin{subequations}
\begin{align}
\langle \Phi, A \chi'(g) \Psi \rangle =
& -i\eta \epslim   \int\, d\theta d\xi F_2\left(\theta +i\epsilon, \xi +i\pi -i\epsilon\right) g^+\left(\xi -\frac{\pi i}{3}\right) \overline{\Phi(\theta)} \Psi\left(\xi +\frac{\pi i}{3}\right) \label{1comma}\\
&-i\eta F_0 \int d\xi\, g^+\left(\xi -\frac{\pi i}{3}\right) \overline{\Phi(\xi)} \Psi\left(\xi +\frac{\pi i}{3}\right); \label{1commb}
\end{align}
\end{subequations}
\begin{subequations}
\begin{align}
\langle \Phi, \chi'(g)A \Psi \rangle =
&-i\eta \epslim \int\, d\xi d\rho\, F_2\left(\xi +\frac{\pi i}{3} -2i\epsilon, \rho +i\pi -i\epsilon\right)g^+\left(\xi -\frac{\pi i}{3}\right) \overline{\Phi(\xi)} \Psi(\rho)\label{1commc}\\
&-i\eta F_0 \int d\xi\, g^+\left(\xi -\frac{\pi i}{3}\right) \overline{\Phi(\xi)}\Psi\left(\xi +\frac{\pi i}{3}\right),\label{1commd}
\end{align}
\end{subequations}
noting that in the expansion \eqref{exp} the coefficients which contribute to the matrix elements $\langle \Phi, A \chi'(g) \Psi \rangle$ and $\langle \Phi, \chi'(g)A \Psi \rangle$ are $F_0$ and $F_2$.
We observe that the $\epsilon$-prescription in the argument of $F_2$ in \eqref{1comma} was introduced in \cite{BC14}, even if $F_2(\pmb{\zeta})$ is actually analytic at $\zeta_2 -\zeta_1 =i\pi$. 
Further, the vector $(A\Psi)_1(\zeta_1) = \int d\zeta_2\, F_2(\zeta_1, \zeta_2) \Psi(\zeta_2)$ in \eqref{1commc} is no longer analytic at $\zeta_1 = \frac{\pi i}{3}$ due to the pole of the integrand at $\zeta_2 -\zeta_1 = \frac{2\pi i}{3}$, so that $\chi'$ as given in Section \ref{chi} cannot be applied to $(A\Psi)_1(\zeta_1)$. However,
by the symmetry of $\chi'(g)$, we can apply it to the vector $\Phi$ and then shifting the integral contour, we obtain the boundary value indicated in \eqref{1commc}.

The terms \eqref{1commb} and \eqref{1commd} depending on $F_0$ cancel each other.

We shift the integral contour in \eqref{1comma} in the variable $\xi$ by $-\frac{i\pi}{3} +i0$ (taking care to not cross the pole hyperplane of $F_2 (\pmb{\zeta})$ at $\operatorname{Im}(\zeta_2 -\zeta_1) = \frac{2\pi}{3}$). For this the integrand is required to decay at infinity which should follow from the growth properties of $g_{\pm}$ and from (P6).
\begin{align*}
 \eqref{1comma} &=-i\eta \epslim \int\, d\theta d\xi F_2\left(\theta +i\epsilon, \xi +i\pi -\frac{i\pi}{3}+2i\epsilon\right) g^+\left(\xi -\frac{2\pi i}{3}\right) \overline{\Phi(\theta)} \Psi(\xi).
\end{align*}
Hence, we find
\begin{subequations}
\begin{align}
\langle \Phi, [A, \chi'(g)] \Psi \rangle =& -i\eta \epslim\int\, d\theta d\xi F_2\left(\theta +i\epsilon, \xi +\frac{2i\pi}{3}+2i\epsilon\right) g^+\left(\xi -\frac{2\pi i}{3}\right) \overline{\Phi(\theta)} \Psi(\xi) \label{Achi1}\\
&+i\eta \epslim\int\, d\xi d\rho\, F_2\left(\xi +\frac{\pi i}{3} -2i\epsilon, \rho +i\pi -i\epsilon\right)g^+\left(\xi -\frac{\pi i}{3}\right) \overline{\Phi(\xi)} \Psi(\rho).\label{Achi2}
\end{align}
\end{subequations}
We consider now the commutator $\langle \Phi, [A, \phi'(g)] \Psi \rangle =\langle \Phi, A z'^\dagger(g^+) \Psi\rangle
-\langle \Phi, z'^\dagger(g^+)A \Psi\rangle 
+\langle \Phi, A z'(J g^-)  \Psi \rangle
-\langle \Phi, z'(Jg^-)A  \Psi\Omega \rangle$, where we inserted the expression of $\phi'(g)$. 

Using the action of \eqref{exp} and of $z',z'^\dagger$ on the one-particle vectors $\Psi, \Phi$, and noting that to these matrix elements the coefficients in \eqref{exp} which contribute are $F_1$ and $F_3$, one computes
\begin{subequations}
\begin{align}
\langle \Phi, [A, \phi'(g)] \Psi \rangle =& 
\epslim \int d\theta d\rho d\xi\, F_3(\theta +i\epsilon, \xi +i\pi -2i\epsilon, \rho +i\pi -i\epsilon) g^+(\xi) \overline{\Phi(\theta)}\Psi(\rho) \label{F3plus}\\
&+ \int d\xi d\rho\, F_1(\xi +i\pi -i0) g^+(\xi) \overline{\Phi(\rho)} \Psi(\rho) \overline{S(\xi -\rho)} \\
&-\epslim \int d\theta d\rho d\xi\, F_3 (\theta +i\epsilon, \xi +2i\epsilon, \rho +i\pi -i\epsilon) g^-(\xi)\overline{\Phi(\theta)}\Psi(\rho) \label{F3minus}\\
&- \int d\rho d\xi\, F_1(\xi +i0) g^-(\xi) \overline{\Phi(\rho)} \Psi(\rho)S(\xi -\rho),
\end{align}
\end{subequations}
where we used that $\overline{(Jg^-) (\xi)}= g^-(\xi)$; for the direction of the boundary values in $F_3$ and $F_1$ see \cite[Definition 4.3]{BC14}.

By shifting the integral in \eqref{F3minus} in the variable $\xi$ from $\mathbb{R}$ to $\mathbb{R} +i\pi$, we find  \eqref{F3plus} up to residues (again bounds at infinity enter). Specifically, the analytically continued function $F_3(\theta +i\epsilon, \xi', \rho +i\pi -i\epsilon)$ has a pole at $\xi' - \theta -i\epsilon = \frac{2\pi i}{3}$ and at $\rho +i\pi -i\epsilon -\xi' = \frac{2\pi i}{3}$. Using \ref{resbou}, the residue at the first pole is given by
\begin{equation*}
\operatorname*{res}_{\xi' -\theta -i\epsilon = \frac{2\pi i}{3}}F_3(\theta +i\epsilon, \xi', \rho +i\pi -i\epsilon)
=- \frac{\eta}{2\pi} F_2 \Big(\theta +i\epsilon +\frac{i\pi}{3}, \rho +i\pi -i\epsilon \Big).
\end{equation*}
This residue gives a contribution to the difference of the two integrals \eqref{F3minus} and \eqref{F3plus}, namely
\begin{equation*}
\eqref{F3plus} + \eqref{F3minus}=- i \eta\int d\theta d\rho\,  F_2\left(\theta +i\epsilon +\frac{\pi i}{3}, \rho +i\pi -i\epsilon\right)g^-\left(\theta +i\epsilon +\frac{2\pi i}{3}\right)\overline{\Phi(\theta)}\Psi(\rho),
\end{equation*}
where the orientation of the residue is fixed using \cite[Lemma 3.4]{BC14}.

In computing the second residue
\begin{multline}
+\operatorname*{res}_{\rho +i\pi -i\epsilon -\xi'= \frac{2\pi i}{3}} F_3(\theta +i\epsilon, \xi', \rho +i\pi -i\epsilon)\\
=
\operatorname*{res}_{\rho +i\pi -i\epsilon -\xi'= \frac{2\pi i}{3}} 
S(\xi' -\theta -i\epsilon)S(\rho +i\pi -i\epsilon -\theta -i\epsilon)F_3(\xi',\rho +i\pi -i\epsilon,\theta +i\epsilon) \\
=\frac{\eta}{2\pi}F_2\left(\theta +i\epsilon, \rho +\frac{2\pi i}{3} -i\epsilon\right),
\end{multline}
we can assume $\theta \neq \rho$ since it suffices to compute a meromorphic function on any open set.
Note that the two $S$-factors are analytic near the pole hyperplane $\rho +i\pi -i\epsilon -\xi'= \frac{2\pi i}{3}$ and only $F_3$ has a pole there. In the second equality we used \ref{resbou} and the Bootstrap equation \ref{bootstrap}. The sign of the residue is fixed by using \cite[Lemma 3.4]{BC14}. 

This residue gives a contribution to the difference of the integrals in \eqref{F3minus} and \eqref{F3plus}, namely
\begin{equation*}
\eqref{F3plus} + \eqref{F3minus} = +i \eta \int d\theta d\rho\, F_2\left(\theta +i\epsilon,\rho +\frac{2i\pi}{3} -i\epsilon\right) g^-\left(\rho +i\pi -i\epsilon -\frac{2\pi i}{3}\right)\overline{\Phi(\theta)}\Psi(\rho).  
\end{equation*}
Summarizing our results, we have
\begin{subequations}
\begin{align}
\langle \Phi, [A,\phi'(g)] \Psi \rangle =& -i\eta  \epslim \int d\theta d\rho\, F_2 \left(\theta +i\epsilon +\frac{\pi i}{3},\rho + i\pi -i\epsilon\right) g^+\left(\theta +i\epsilon -\frac{\pi i}{3}\right)\overline{\Phi(\theta)}\Psi(\rho) \label{Aphi1}\\
&+i\eta  \epslim\int d\theta d\rho\, F_2\left(\theta +i\epsilon,\rho +\frac{2\pi i}{3}-i\epsilon\right) g^+\left(\rho -i\epsilon -\frac{2\pi i}{3}\right)\overline{\Phi(\theta)} \Psi(\rho)\label{Aphi2}\\
&+ \int d\xi d\rho\, F_1(\xi +i\pi -i0) g^+(\xi) \overline{\Phi(\rho)} \Psi(\rho) \overline{S(\xi -\rho)} \label{F1}\\
&- \int d\rho d\xi\, F_1(\xi +i0) g^-(\xi) \overline{\Phi(\rho)} \Psi(\rho)S(\xi -\rho) \label{F1two},
\end{align}
\end{subequations}
where we used that $g^+(\theta) = g^-(\theta \pm i\pi)$.

In \eqref{Achi1} the coefficient $F_2(\pmb{\zeta})$ approaches the pole in the direction $\zeta_2 -\zeta_1 = \frac{2\pi i}{3} +i\epsilon$, instead in \eqref{Aphi2} it approaches the pole in the direction $\zeta_2 -\zeta_1 = \frac{2 \pi i}{3} -i\epsilon$. So the difference of these boundary values is again a residue that we can compute using \ref{resbou} and \cite[Lemma 3.4]{BC14},
\begin{equation*}
 \operatorname*{res}_{\xi -\theta =0} F_2 \left(\theta,\xi +\frac{2\pi i}{3}\right) = +\frac{\eta}{2\pi} F_1 \left(\theta +\frac{\pi i}{3} \right).
\end{equation*}
This residue gives a contribution to the difference of \eqref{Achi1} and \eqref{Aphi2}, namely
\begin{equation}\label{f1}
\eqref{Achi1} + \eqref{Aphi2}= +(i\eta)^2 \int d\theta\, F_1\left(\theta +\frac{i\pi}{3}\right)g^+\left(\theta -\frac{2\pi i}{3}\right)\overline{\Phi(\theta)}\Psi(\theta).
\end{equation}
Likewise, the difference of the boundary values in the integrals \eqref{Achi2} and \eqref{Aphi1} gives the residue
\begin{equation*}
 \operatorname*{res}_{\rho -\xi =0} F_2\left(\xi +\frac{\pi i}{3},\rho +i\pi\right)= -\frac{\eta}{2\pi} F_1 \left(\rho +\frac{2\pi i}{3}\right),
\end{equation*}
yielding
\begin{equation}\label{f1two}
\eqref{Achi2} +\eqref{Aphi1}= -(i\eta)^2 \int d \rho\, F_1 \left(\rho +\frac{2\pi i}{3}\right)g^+\left(\rho -\frac{\pi i}{3}\right)\overline{\Phi(\rho)}\Psi(\rho). 
\end{equation}
We now consider \eqref{F1} and \eqref{F1two}. When shifting the integral contour in \eqref{F1two} in the variable $\xi$ from $\mathbb{R}$ to $\mathbb{R} +i\pi$, $F_1$ is analytic, but the factor $S(\zeta)$ has poles at $\zeta =\frac{2\pi i}{3}, \frac{\pi i}{3}$, the residues of which we denote by $R,R'$, respectively. Using \cite[Equation (3.14)]{BC14}, the difference of the boundary values in the integrals \eqref{F1} and \eqref{F1two} gives the following residue
\begin{multline}\label{f1finalres}
 \eqref{F1} +\eqref{F1two}
= -\big( \operatorname*{res}_{\xi -\rho = \frac{2\pi i}{3}} + \operatorname*{res}_{\xi -\rho = \frac{\pi i}{3}} \big) \int d\rho\, F_1(\xi) g^-(\xi)\overline{\Phi(\rho)}\Psi(\rho)S(\xi -\rho) \\
 =-2\pi i R \int d\rho\, F_1\left(\rho +\frac{2\pi i}{3}\right) g^-\left(\rho +\frac{2\pi i}{3}\right)\overline{\Phi(\rho)}\Psi(\rho) \\
   -2\pi i R' \int d\rho\,  F_1\left(\rho +\frac{\pi i}{3}\right)g^-\left(\rho +\frac{\pi i}{3}\right)\overline{\Phi(\rho)}\Psi(\rho). 
\end{multline}
Comparing this with \eqref{f1two}  and \eqref{f1}, we find
\begin{equation*}
-\eqref{f1finalres} = \eqref{f1two} +\eqref{f1},
\end{equation*}
where we used that $R'=-R$, $\eta =i\sqrt{2\pi |R|}$, $R =i |R|$ and  $g^+(\rho) = g^-(\rho \pm i\pi)$.

Hence, the matrix element $\langle \Phi, [A, \fct'(g)] \Psi \rangle$ vanishes.

\subsection{A guesswork for the bound-state operator}\label{guess}
Our operator $\chi(f)$ appears to be new and it is worth presenting how one can ``find'' it
by formal arguments. Again in this subsection we ignore all subtleties with domains
and assume that we can compute everything termwise.

Borchers, Buchholz and Schroer proved in \cite{BBS01} that, for any Haag-Kastler net with mass eigenvalues,
one can construct polarization-free generators, namely (unbounded) operators which create one-particle
states from the vacuum and are affiliated to the von Neumann algebra associated to wedges.
In particular, if we could construct a Haag-Kastler net for a given scattering function $S$,
then there would have to be polarization-free generators. On the other hand, the expansion (\ref{exp})
in Section \ref{compatibility} should exist, possibly in a modified form.
Let us assume that $B$ is a polarization free generator.
As $B$ generates only a one-particle state from the vacuum, by a (termwise) straightforward computation,
one realizes that $f_{m,0}^{[B]}$ must vanish except $m=0,1$. We may assume that $B$ is symmetric,
then also $f_{0,n}^{[B]}$ survive only if $n=0,1$.

Let $h$ be a real test function such that $\supp h \subset W_\L - (0,r)$.
The field $\phi(h) = z^\dagger(h^+) + z(J_1h^-)$ of \cite{Lechner03}
has only $(m,n) = (0,1), (1,0)$ components.
As we know that this cannot be wedge-local, let us consider the simplest variation of it.
The $(0,0)$ component is a scalar and has no effect on the commutation relations.
Assume $B$ has a $(1,1)$ component, namely that $B = z^\dagger(h^+) + X + z(J_1h^-)$, that $X$ preserves the one-particle space and $X\Omega = X^*\Omega = 0$.

Suppose that we have the Haag-Kastler net. If $B$ is affiliated to $\A(W_\L - (0,r))$,
then it must commute with operators $A$ which arise from form factors as in Section \ref{compatibility}.
Especially, for a one-particle vector $\Phi$, the commutator
$\<\Omega, [A,B]\Phi\>$ must vanish. We may assume that both $A$ and $B$ are symmetric,
and thus $h^+ = J_1h^-$.
The vector $\Omega$ is annihilated by $z(J_1h^-)$ and $X$, therefore, the only remaining terms
are
\begin{align*}
 \<\Omega, [A,B]\Phi\> &= \<\Omega, Az^\dagger(h^+)\Phi\> + \<A\Omega, X\Phi\> - \<z^\dagger(h^+)\Omega, A\Phi\> \\
 =& \int d\theta_1d\theta_2\, \left(f_{0,2}^{[A]}(\theta_1,\theta_2)h^+(\theta_2)\Phi(\theta_1) 
 - f_{1,1}^{[A]}(\theta_2,\theta_1)\overline{h^+(\theta_2)}\Phi(\theta_1)\right) + \<A\Omega, X\Phi\> \\
 =& \int d\theta_1d\theta_2\, \big(F_2^{[JA^* J]}(\theta_1 +i0,\theta_2 +i0)h^+(\theta_2)\\
  &\;\;\;\;- F_2^{[JA^* J]}(\theta_1 + i0,\theta_2 + i\pi - i0)h^+(\theta_2 + \pi i)\big)\Phi(\theta_1) 
  + \<A\Omega, X\Phi\> \\
 =& \int d\theta_1\; i\eta F_1^{[JA^*J]}\left(\theta_1 + \frac{\pi i}3\right)h^+\left(\theta_1 + \frac{2\pi i}3\right)\Phi(\theta_1) + \<A\Omega, X\Phi\> \\
 =& \int d\theta_1\; i\eta F_1^{[JA^* J]}(\theta_1)h^+\left(\theta_1 + \frac{\pi i}3\right)\Phi\left(\theta_1 - \frac{\pi i}3\right) + \<A\Omega, X\Phi\> \\
 =& \int d\theta_1\; i\eta F_1^{[A]}(\theta_1 + i\pi - i0)h^+\left(\theta_1 + \frac{\pi i}3\right)\Phi\left(\theta_1 - \frac{\pi i}3\right) + \<A\Omega, X\Phi\> \\
 =& -\<A\Omega, \chi(h)\Phi\> + \<A\Omega, X\Phi\>,
\end{align*}
where in the 3rd equality we used $S$-periodicity of $F_2$, and the relations between form factors of $A$ and $JA^*J$
in \cite[Proposition 8.9]{Cadamuro12} (which is also used in the 6th equality).
In the 4th equality, we used the residue formula \ref{resbou}  and then in the 5th equality we shifted the
integral contour by $-\frac{\pi i}3$, where we assumed that $\Phi$ does not have poles
(again, this is a guesswork).
Now, if we assume the Reeh-Schlieder property for the net $\A$, there must be sufficiently many
local observables $A$, in particular, there are many such $A$'s that
$\H_1$ is spanned by the one-particle components of $A\Omega$'s.
As we require that the above commutator should vanish, it must hold that $X\Phi = \chi(h)\Phi$.
Namely, we correctly guessed the one-particle action of $\chi(h)$.
We know that this choice works, as we saw in Section \ref{field-scalar}.

Note that $\chi(h)$ cannot be written exactly in the form of (\ref{exp}), due to its subtle
domain property. Recall \cite{Cadamuro12} that the expansion (\ref{exp}) is available only
for certain regular operators.
We were fortunate that this guesswork resulted in a correct answer,
but in more complicated models it fails \cite{CT15-2}.
On the other hand, if the Haag-Kaslter net for the model existed at all, there would have to be polarization-free generators
\cite{BBS01}. This failure might be a consequence of the complicated domains, or one would have to
add higher component in the expansion (\ref{exp}).
After all, we are not excluding wedge-local fields which are more complicated, yet have
better domain properties.

\section{Conclusions and outlook}\label{sec:conclusions}

In the present work, our aim was to extend Lechner's construction of two-dimensional models of
quantum field theory with a factorizing scattering matrix to models with bound states,
which are associated with poles of the S-matrix in the physical strip.
We considered scalar S-matrices with only one pair of poles in the physical strip,
corresponding to two bosons of the same species which fuse into another boson of the same species.

The properties of the new class of S-matrices include the presence of these extra poles in
the physical strip (and the values of the corresponding residues of the S-matrix at these points)
and the Bootstrap equation, in addition to the properties that can be found in \cite{Lechner:2006}.
Quantum integrable models with underlying S-matrix fulfilling this new set of properties include
the Bullough-Dodd model.

For this class of S-matrices we constructed wedge-local fields $\fct(f)$ by adding the
bound-state operator $\chi$ to the wedge-local field of Lechner. 
However, with the simple domain taken here, the field $\fct(f)$ is expected not to be
self-adjoint and we have not shown strong commutativity.
Rather, we have shown that the field $\fct(f)$ weakly commute with its reflected field $\fct'(g)$
on a common domain.

Furthermore, we gave a partial answer to the question of compatibility of our construction
with the form factor program. In particular, we conjectured that a suitable variation of the
characterization theorem for local observables in \cite{BC14} holds in our class of models:
we argued explicitly that this is true at least in one-particle states.

\subsubsection*{Open problems}
One of the major open problems is to show the strong commutativity of the fields
$\fct(f)$ and $\fct'(g)$ in order to obtain a Borchers triple.
This would mean not only proving the existence of self-adjoint extensions
of the two fields, but also selecting extensions that strongly commute.
This is a highly non-trivial task for which the results available in the literature can
offer only partial answers. Some results on the construction of extensions have already
been obtained by one of the authors \cite{Tanimoto15-1}.

Following Buchholz and Lechner \cite{BL04}, the natural next step would then be to show the
Bisognano-Wichmann property, i.e., the geometric action of the modular group.
Thereafter, one would like to establish modular nuclearity condition \cite{Lechner08}
in order to prove that the wedge algebras so generated are split,
and therefore the non-triviality of the double-cone algebras.

Another interesting problem would be an extension of our construction, valid for the
moment only for scalar S-matrices, to a larger class of integrable models with bound
states, such as the $Z(N)$-Ising model \cite{BFK06} and the sine-Gordon model \cite{ZZ79}
for a certain range of the coupling constant.
Indeed, in the $Z(3)$
model one can find a multi-component field which fulfills weak wedge-commutativity at
least on the level of formal computation \cite{CT15-2}.
In the $Z(N)$ (with $N > 3$) and sine-Gordon models,
this will hold only for certain components of the fields.
The construction involves a multi-component generalization of the bound-state operator $\chi$,
which will now modify the multi-component field of Lechner-Sch\"utzenhofer \cite{LS14}.
The matrix-valued analogue of the properties of $S$ introduced in Section \ref{sec:scalar}
will need to include the Yang-Baxter equation, which will play an important role in the proof
of wedge-commutativity.

\subsubsection*{Acknowledgement}
We thank Marcel Bischoff, Henning Bostelmann, Wojciech Dybalski, Gian Michele Graf, Gandalf Lechner,
Valter Moretti and Bert Schroer for interesting discussions, and Luca Giorgetti for useful comments.

We benefitted from the opportunity during the workshop
``Algebraic Quantum Field Theory: Its status and its future'' at Erwin Schr\"odinger Institute, Vienna.
This work was supported by Grant-in-Aid for JSPS fellows 25-205.

\appendix
\section{Classification of scalar S-matrices}\label{classification}
Here we classify all the functions $S(\zeta)$ which satisfy the conditions \ref{unitarity}--\ref{poles-boundedness}
in Section \ref{scalar-s}.
Let us take such an $S$. By \ref{poles-boundedness}, $S$ has simple poles at $\zeta = \frac{\pi i}3, \frac{2\pi i}3$.
Recall the function
\[
f_\frac23(\zeta) = -\frac{\sinh\frac12\left(\zeta + \frac{\pi i}3\right)}{\sinh\frac12\left(\zeta - \frac{\pi i}3\right)}
  \frac{\sinh\frac12\left(\zeta + \frac{2\pi i}3\right)}{\sinh\frac12\left(\zeta - \frac{2\pi i}3\right)},
\]
which satisfies \ref{unitarity}--\ref{bootstrap}. It has simple poles exactly at $\zeta = \frac{\pi i}3, \frac{2\pi i}3$,
and is bounded below in $\RR + i[0,\pi]$.
Let us consider
\[
 \underline S(\zeta) = \frac{S(\zeta)}{f_\frac23(\zeta)}.
\]
Now this $\underline S$ satisfies \ref{unitarity}--\ref{bootstrap},
is analytic and bounded in $\RR + i(0,\pi)$.
Therefore, by the exponential map from $\RR + i(0,\pi)$ onto the unit disk in $\CC$,
it follows that $\underline S$ is an inner function and
admits the factorization $\underline S = c\cdot S_\infty \cdot S_\blaschke$
into a constant $c$ with $|c| = 1$, the singular factor $S_\infty$
which has no zero and the Blaschke factor $S_\blaschke$
\cite[Theorem 17.15]{Rudin87}. Both of them satisfy \ref{unitarity}.
See also \cite[Appendix A]{LW11} and \cite[Section 1]{LST13} for its description in $\RR + i(0,\pi)$.

Let us look at the bootstrap equation \ref{bootstrap},
where the boundary value is represented as $S(\theta) = \frac{S(\theta + \frac{\pi i}3)}{S(\theta + \frac{2\pi i}3)}$.
The right-hand side is meromorphic, therefore, $S$ must have a meromorphic continuation
in a neighborhood of $\RR$ as well.
Yet, by \ref{unitarity}, the boundary value has modulus $1$, hence it is actually continuous
in a neighborhood of $\RR$.
In general, an inner function can have essential singularities at the boundary.
It follows from the observations above that the only possible essential singularities
in the picture of $\RR + i(0,\pi)$ is $\theta = \pm \infty$, as in \cite[Appendix A]{LW11}.

By the uniqueness of the factorization, the properties \ref{hermitian}--\ref{crossing}
must be satisfied by the constant, the singular and the Blaschke factors separately.
Therefore, $c = \pm 1$.
Consider next the singular part $S_\infty$. It admits an integral representation
as \cite[Theorem 17.15]{Rudin87} and it has essential singularities at
points in the support of the measure.
By considering the symmetry \ref{hermitian} and the possible essential singularities as above,
it has the form $S_\infty(\theta) = e^{ia(e^\theta - e^{-\theta})}, a \ge 0$.
By a straightforward computation, this satisfies \ref{bootstrap} and $S_\infty(\frac{\pi i}3) = S_\infty(\frac{2\pi i}3) > 0$.

Let us then turn to the Blaschke factor $S_\blaschke$.
It can be written as the following infinite product:
\[
 S_\blaschke(\zeta) = \prod_n c_n\frac{e^\zeta - e^{\a_n}}{e^\zeta - e^{\overline{\a_n}}},
\]
where $\a_n \in \RR + i(0,\pi)$ and $c_n = -\frac{|\b_n|}{\b_n}\frac{1-\b_n}{1-\overline{\b_n}}$,
is a constant, in accordance with \cite[Theorem 15.21]{Rudin87}, where $\b_n = \frac{e^{\a_n}-i}{e^{\a_n}+i}$
(and $c_n = 1$ if $\a_n = \frac{\pi i}2$ (and hence $\beta_n = 0$) by convention).
These $\a_n$'s are exactly the zeros of $S_\blaschke$.
By \ref{hermitian}, $\a_n$ is purely imaginary or it must appear in a pair with
$-\overline{\a_n}$ (including multiplicity). Next,
by \ref{crossing}, $\im \a_n = \frac\pi 2$ or it must appear in a pair with $i\pi - \a_n$.
In each case, it is straightforward to see that the constant factors $c_n$ cancel each other
and $S_\blaschke(\zeta) = \prod_n \frac{e^\zeta - e^{\a_n}}{e^\zeta - e^{\overline{\a_n}}}$.

As $\underline S$ and $S_\infty$ satisfy the bootstrap equation \ref{bootstrap},
so does $c\cdot S_\blaschke$.
Therefore, if there is any $\a_n$ such that $\im \a_n < \frac{\pi}3$,
then $0 < \im \a_n + \frac{\pi}3 < \pi$ and $\a_n + \frac{\pi i}3$ must be another zero.
If $\frac{\pi}3 < \im \a_m$, then by combining \ref{crossing} and \ref{bootstrap},
we may assume that there is an $\a_n$ such that $\im \a_n < \frac{\pi}3$.

We summarize these observations. Let $\alpha_n$ with $\re \a_n \neq 0$ appear in the product.
Then, if $\im \a_n = \frac{\pi}3$ (if $\im \a_m = \frac{2\pi}3$, by \ref{crossing} we may assume
that there is $\a_n$ such that $\im \a_n = \frac{\pi}3$), then the factor
\[
  \frac{(e^\zeta - e^{\a_n})(e^\zeta - e^{-\overline{\a_n}})(e^\zeta - e^{\pi i -\a_n})(e^\zeta - e^{\pi i + \overline{\a_n}})}{(e^\zeta - e^{\overline{\a_n}})(e^\zeta - e^{-\a_n})(e^\zeta - e^{\pi i - \overline{\a_n}})(e^\zeta - e^{\pi i + \a_n})}
\]
must appear. If $\im \a_n \neq \frac{\pi}3$, then we may assume that $\im \a_n < \frac{\pi}3$ and
\begin{align*}
  & \frac{(e^\zeta - e^{\a_n})(e^\zeta - e^{-\overline{\a_n}})(e^\zeta - e^{\pi i -\a_n})(e^\zeta - e^{\pi i + \overline{\a_n}})}{(e^\zeta - e^{\overline{\a_n}})(e^\zeta - e^{-\a_n})(e^\zeta - e^{\pi i - \overline{\a_n}})(e^\zeta - e^{\pi i + \a_n})} \times \\
  & \frac{(e^\zeta - e^{\a_n + \frac{\pi i}3})(e^\zeta - e^{-\overline{\a_n +\frac{\pi i}3}})(e^\zeta - e^{\pi i - (\a_n+\frac{\pi i}3)})(e^\zeta - e^{\pi i + \overline{\a_n+\frac{\pi i}3}})}{(e^\zeta - e^{\overline{\a_n - \frac{\pi i}3}})(e^\zeta - e^{-(\a_n+\frac{\pi i}3)})(e^\zeta - e^{\pi i - \overline{\a_n+\frac{\pi i}3}})(e^\zeta - e^{\pi i + \a_n+\frac{\pi i}3})}
\end{align*}
must appear. Both of these factors satisfy \ref{unitarity}--\ref{crossing}, and also \ref{bootstrap}
by straightforward computations. Furthermore, these factors take positive numbers if $\zeta$ is purely imaginary:
each of these products satisfies \ref{hermitian}, hence it takes non-zero real value on the imaginary axis,
therefore, it has a fixed sign, but at $\zeta = 0$ it can be directly checked that it takes $1$,
hence must take positive values on the imaginary axis.

Finally, if $S_\blaschke$ has a zero at $\a_n$ such that $\re \a_n = 0$,
then the factors $f_{\frac{B_k}3 - \frac23}(\zeta) f_{-\frac{B_k}3}(\zeta)$ of Section \ref{scalar-s}
must appear.
These factors satisfies \ref{unitarity}--\ref{bootstrap}, and
$f_{\frac{B_k}3 - \frac23}(\frac{\pi i}3)f_{-\frac{B_k}3}(\frac{2\pi i}3) = f_{\frac{B_k}3 - \frac23}(\frac{2\pi i}3) f_{-\frac{B_k}3}(\frac{2\pi i}3)< 0$.
Now, as the Blaschke factor is determined by zeros, $S_\blaschke$ satisfies \ref{bootstrap} by itself,
which implies that $c$ also must satisfy \ref{bootstrap}, namely $c=1$.
In order to assure \ref{poles-boundedness}, $S_\blaschke$ can have an odd number of
factors as $f_{\frac{B_k}3 - \frac23}(\zeta) f_{-\frac{B_k}3}(\zeta)$
(infinite is impossible, because $S(\zeta)$ would have an essential singularity
at $\zeta = 0$, which we have excluded), so that the residues of the factor $f_\frac23$
gets multiplied by a negative number in the whole two-particle S-matrix
$S(\zeta) = f_\frac23(\zeta)\cdot S_\infty(\zeta) S_\blaschke(\zeta)$
(recall that $f_\frac23$ has only simple poles).
Now it is obvious that $S_\blaschke(0) = 1$.

Finally, a general $S$ is of the form
\[
 S(\zeta) = f_\frac23(\zeta)\cdot e^{ia(e^\zeta - e^{-\zeta})}S_\blaschke(\zeta),
\]
where the Blaschke factor $S_\blaschke$ contains only combinations specified above
and should satisfies the Blascke condition.

Now, as one has $f_\frac23(0) = -1, e^{ia(e^0 - e^0)} = 1, S_\blaschke(0) = 1$,
it holds that $S(0) = -1$, namely we have \ref{fermionic}.

\section{Lemmas on shifting integral contour}
Here we state some easy facts explicitly for the sake of clarity.

\begin{lemma}\label{lm:abstractshift}
 Let $A_1$ and $A_2$ be positive self-adjoint operators which strongly commute.
 If $\Psi, \Phi \in \dom(A_1)\cap\dom(A_2)$, then for any $0\le \epsilon\le 1$, we have that
 $\Psi, \Phi \in \dom(A_1^{1-\epsilon} A_2^\epsilon)$ and it holds that
 \[
  \<A_1\Phi, A_2\Psi\> = \<A_1^\epsilon A_2^{1-\epsilon}\Phi, A_1^{1-\epsilon}A_2^\epsilon\Psi\> .
 \]
\end{lemma}
This might seem obvious, and indeed the proof is not too complicated, yet
one should note that $A_1A_2\Psi$ etc. are ill-defined.
\begin{proof}
 Consider the joint spectral decomposition with respect to $A_1$ and $A_2$.
 Any vector $\Psi$ is a (possibly vector-valued) $L^2$-function on this spectral space of two variables
 $a_1, a_2$, such that $A_1$ and $A_2$ act by multiplication with $e^{a_1}, e^{a_2}$, respectively.
 
 Now, the hypothesis that $\Psi \in \dom(A_1)$ means that $\Psi(a_1,a_2)$ and $e^{a_1}\Psi(a_1,a_2)$ are
 both $L^2$. This implies that $e^{2(1-\epsilon) a_1}|\Psi(a_1,a_2)|^{2(1-\epsilon)}$
 is $L^\frac1{1-\epsilon}$. 
 Similarly, if $\Psi \in \dom(A_2)$, then $e^{a_2}\Psi(a_1,a_2)$ is $L^2$ and hence,
 $e^{2\epsilon a_2}|\xi(a_1,a_2)|^{2\epsilon}$ is $L^{\frac1\epsilon}$.

 The first claim of the present lemma is that $\Psi \in \dom(A_1^{1-\epsilon}A_2^\epsilon)$, which is equivalent to
 $e^{2(1-\epsilon)a_1 + 2\epsilon a_2}|\Psi(a_1,a_2)|^2$ being $L^1$.
 This is a consequence of H\"older's inequality and of the estimates above.
 
 The desired equality is obtained simply by writing it as the spectral integral:
 \begin{align*}
  \<A_1\Phi, A_2\Psi\> &= 
  \int da_1da_2\, e^{a_1}\overline{\Phi(a_1,a_2)}e^{a_2}\Psi(a_1,a_2) \\
  &= \int da_1da_2\, e^{\epsilon a_1 + (1-\epsilon)a_2}\overline{\Phi(a_1,a_2)}e^{(1-\epsilon)a_1 + \epsilon a_2}\Psi(a_1,a_2) \\
  &=  \<A_1^\epsilon A_2^{1-\epsilon}\Phi, A_1^{1-\epsilon}A_2^\epsilon\Psi\>.
 \end{align*}
\end{proof}

\begin{lemma}\label{lm:l2shift}
 Let $\K$ be a Hilbert space.
 Let $\epsilon > 0$ and $\Psi, \Phi$ be $\K$-valued analytic functions in $\RR + i(-\epsilon,0)$
 such that for $-\epsilon < \a < 0$,  $\Psi(\theta + i\a), \Phi(\theta + i\a)$ are
 $\K$-valued $L^2$-functions in $\theta$, uniformly bounded in $\a$.
 Then it holds that
 \[
  \int d\theta\, \<\Phi(\theta - \epsilon i), \Psi(\theta)\> =
  \int d\theta\, \<\Phi(\theta), \Psi(\theta - \epsilon i)\>.
 \]
\end{lemma}
\begin{proof}
 $\Phi$ and $\Psi$ can be considered as vectors in $L^2(\RR,d\theta)\otimes \K$
 and in the domain of the operator $A\otimes \1$, where $(A^{it}\xi)(\theta) = \xi(\theta + t\epsilon)$.
 One can either apply Lemma \ref{lm:abstractshift}, or essentially repeat the proof of Proposition \ref{pr:chi:symmetry}
 by scaling the variable and putting a constant instead of $f^+$.
\end{proof}

\section{Comments on the domains}\label{domain}
We have seen that, if $f$ is a real test function
supported in $W_\L$, then $\chi(f)$ is a symmetric operator.
Then the question arises whether $\chi(f)$ is (essentially) self-adjoint and, if not,
what the appropriate extension is.

\subsubsection*{There are many self-adjoint extensions}
We claim that the answer to the first question is in general no. 
Indeed, if one considers a similar operator which is a product of a multiplication operator
by an analytic function and an imaginary shift, then its properties depend very much on the 
analytic function. More precisely, the deficiency indices of the product operator
are tightly related to the zeros of the analytic function and its decay rate as $|\re \zeta| \to \infty$
\cite{Tanimoto15-1}.
In our case, as $f^+$ is strongly decreasing, we expect that $\chi_1(f) = x_f \Delta^{\frac16}$
should have deficiency indices $(\infty,\infty)$.

It is actually very easy to see that there are more than one self-adjoint extension
if $f$ is of a special form. Namely, let $f = h * \bar h$ be a convolution.
Then, after simple computations, one obtains
\begin{align*}
 \chi_1(f) &\subset \sqrt{2\pi|R|} \cdot x_h\Delta^\frac1{12}\cdot \Delta^\frac1{12}x_h^*, \\
 \chi_1(f) &\subset \sqrt{2\pi|R|} \cdot \Delta^\frac1{12}y_h\cdot y_h^*\Delta^\frac1{12},
\end{align*}
where $x_h$ is defined in Proposition \ref{pr:chi:symmetry} and $y_h$ is the multiplication operator by the function $h^+(\theta + \frac{\pi i}2)$.
It is easy to see that $\Delta^\frac1{12}x_h^*$ and $\Delta^\frac1{12}y_h^*$
are densely defined and closed, therefore $x_h\Delta^\frac1{12}$ and $y_h\Delta^\frac1{12}$
are closable. For any closed operator $X$, $X^*X$ is self-adjoint by a theorem of von Neumann
\cite[Theorem X.25]{RSII}. Therefore, we have two extensions
$\overline{x_h\Delta^\frac1{12}}\cdot \Delta^\frac1{12}x_h^*$ and
$\Delta^\frac1{12}y_h\cdot \overline{y_h^*\Delta^\frac1{12}}$, which are self-adjoint due
to the theorem of von Neumann. One can see that they are indeed different extensions when
$h^+$ has zeros in the strip.

Self-adjoint extensions of $\chi_n(f)$ are more complicated. One can find extensions using
the above decomposition of $\chi_1(f)$ and again the theorem of von Neumann,
but it is hard to establish strong commutativity between $\fct(f)$ and $\fct'(g)$.

Therefore, the situation is quite different from that of \cite{Lechner03}.
There, the wedge-local field $\phi(f)$ (for an analytic $S$, hence without the $\chi(f)$ term)
is bounded on each subspace of a fixed particle number,
essentially self-adjoint on the space of finite particle number and bounded by the Hamiltonian.
In particular, there is only one self-adjoint extension, which is the closure of $\phi(f)$ defined
on the space of finite particle number.

\subsubsection*{Standard tools fail}
In our case, $\chi(f)$ is already unbounded on each of $\H_n$. It is neither bounded by
the Hamiltonian, therefore one cannot use Nelson's commutator theorem with it \cite[Theorem X.36]{RSII}.
The imaginary shift operator $\Delta$ (or some powers of it) can bound $\chi_1(f)$,
but it cannot bound the commutator. Indeed, the self-adjoint domain should highly depend
on $f$, therefore the commutator theorem which would imply essential self-adjointness
on the domain of a common operator should not be used here.

Another standard tool to prove self-adjointness is the analytic vector theorem of Nelson
\cite[Theorem X.39]{RSII}. Again, this theorem cannot be used here because it would prove
essential self-adjointness, while the right self-adjoint domain of $\chi_1(f)$ should
depend on $f$. But it is interesting to see that the vacuum vector $\Omega$
is not an analytic vector for $\fct(f) = \phi(f) + \chi(f)$. Indeed, $\fct(f)\Omega = f^+ \in \H_1$.
But by definition $f^+$ has a double-exponentially
diverging continuation on the lower strip, and the second application of $\fct(f)$
which contains $\chi(f)$ requires such a continuation.
Although $\chi(f)$ contains the multiplication by $f^+(\cdot + \frac{i\pi}3)$,
altogether it is generically not $L^2$. The third application of $\fct(f)$ is even worse.
In other words, $\Omega$ is not in the naive domain of $\fct(f)^2$.

As $\chi(f)$ is generically not essentially self-adjoint on a naive domain,
there is no hope to apply simply perturbation results, e.g.\! the Kato-Rellich theorem \cite[Theorem X.12]{RSII},
which would imply that the perturbed operator has again the same domain of self-adjointness of
the unperturbed operator under control.

\subsubsection*{Different extensions correspond to different physics}

As remarked above, it is not difficult to find a single self-adjoint extension of $\chi(f)$.
Thereafter, one can show that $\fct(f)$ is essentially
self-adjoint on a simple domain\footnote{We owe this argument to Henning Bostelmann.}.
However, proving that $\fct(f)$ and $\fct'(g)$ strongly commute turns out to be hard.
Indeed, as the operators $\fct(f)$ and $\fct'(g)$ depend on the chosen extensions,
we have to choose right extensions so that the two extensions strongly commute.

Neither is this problem trivial, nor should it be dismissed as
``technical''. The problem of self-adjoint extensions appears also
in the study of differential operators. It is very well known (e.g.\! \cite[Section X.1, Example 2]{RSII})
that different self-adjoint extensions of a same differential operator may correspond to different
boundary conditions, therefore to different systems and different physics, depending
on interpretations.

We will need self-adjoint extensions of $\fct(f)$ and $\fct'(g)$ which strongly commute.
As the strong commutativity depends on the chosen extensions, finding the right self-adjoint
extensions of $\chi(f)$ and $\chi'(g)$ is unavoidable.
By comparing with the situation of differential operators, it is reasonable to believe that
the problem of several self-adjoint extensions, which does not appear in the cases with analytic S-matrices,
is rooted in the essential properties of bound states, therefore worth a serious investigation.

We plan to systematically study this issue in \cite{Tanimoto15-1}.

\end{document}